\documentclass[11pt]{article}
\usepackage[T1]{fontenc}
\usepackage{fullpage,comment} 
\usepackage[utf8]{inputenc}
\usepackage{amsmath}
\usepackage{amsthm}
\usepackage{amssymb} 
\usepackage{mathtools,xparse} 
\usepackage{thm-restate}
\usepackage{xargs}
\usepackage{bbm}
\usepackage{bbold}
\usepackage{hyperref}
\usepackage[pdftex,dvipsnames]{xcolor}
\usepackage[colorinlistoftodos,prependcaption,textsize=tiny]{todonotes}
 \usepackage{algorithm} 

\usepackage{physics}
\usepackage{amsmath}
\usepackage{tikz}
\usepackage{mathdots}
\usepackage{yhmath}
\usepackage{cancel}
\usepackage{color}
\usepackage{array}
\usepackage{multirow}
\usepackage{amssymb}
\usepackage{tabularx}
\usepackage{extarrows}
\usepackage{booktabs}

\usetikzlibrary{fadings}
\usetikzlibrary{patterns}
\usetikzlibrary{shadows.blur}
\usetikzlibrary{shapes}

\usepackage{algorithm}
\PassOptionsToPackage{noend}{algpseudocode}
\usepackage{algpseudocode}

\errorcontextlines\maxdimen

\makeatletter
    \newcommand*{\algrule}[1][\algorithmicindent]{\makebox[#1][l]{\hspace*{.5em}\thealgruleextra\vrule height \thealgruleheight depth \thealgruledepth}}%
\newcommand*{\thealgruleextra}{}
\newcommand*{\thealgruleheight}{.75\baselineskip}
\newcommand*{\thealgruledepth}{.25\baselineskip}

\newcount\ALG@printindent@tempcnta
\def\ALG@printindent{%
    \ifnum \theALG@nested>0
        \ifx\ALG@text\ALG@x@notext
        \else
            \unskip
            \addvspace{-1pt}
            \ALG@printindent@tempcnta=1
            \loop
                \algrule[\csname ALG@ind@\the\ALG@printindent@tempcnta\endcsname]%
                \advance \ALG@printindent@tempcnta 1
            \ifnum \ALG@printindent@tempcnta<\numexpr\theALG@nested+1\relax
            \repeat
        \fi
    \fi
    }%
\usepackage{etoolbox}
\patchcmd{\ALG@doentity}{\noindent\hskip\ALG@tlm}{\ALG@printindent}{}{\errmessage{failed to patch}}
\makeatother

\newbox\statebox
\newcommand{\myState}[1]{%
    \setbox\statebox=\vbox{#1}%
    \edef\thealgruleheight{\dimexpr \the\ht\statebox+1pt\relax}%
    \edef\thealgruledepth{\dimexpr \the\dp\statebox+1pt\relax}%
    \ifdim\thealgruleheight<.75\baselineskip
        \def\thealgruleheight{\dimexpr .75\baselineskip+1pt\relax}%
    \fi
    \ifdim\thealgruledepth<.25\baselineskip
        \def\thealgruledepth{\dimexpr .25\baselineskip+1pt\relax}%
    \fi
    \State #1%
    \def\thealgruleheight{\dimexpr .75\baselineskip+1pt\relax}%
    \def\thealgruledepth{\dimexpr .25\baselineskip+1pt\relax}%
}

\usepackage{color}
\usepackage{tikz}
\usetikzlibrary{shapes.geometric, arrows,automata} 
\usepackage{graphicx}
\usepackage[normalem]{ulem}
\usepackage{breqn}

\DeclareMathOperator*{\argmin}{arg\,min}

\DeclareMathOperator{\SAL}{ENERGY} 
\DeclareMathOperator{\TIME}{TIME}     

\newcommand{\SAF}{\mathrm{SAF}}
\newcommand{\WARN}{\mathrm{WARN}}
\DeclareMathOperator{\randclust}{\mathcal{C}}

\newcommand{\UPCAST}{\mathtt{UPCAST}}
\newcommand{\DOWNCAST}{\mathtt{DOWNCAST}}
\newcommand{\INTERCAST}{\mathtt{INTERCAST}}

\newcommand{\LOCAL}{\mathsf{LOCAL}}
\newcommand{\CONGEST}{\mathsf{CONGEST}}
\newcommand{\RADIO}{\mathsf{RADIO}\text{-}\mathsf{CONGEST}}

\newcommand{\R}{\mathbb{R}}
\newcommand{\ignore}[1]{}

\newcommand{\ID}{\mathsf{ID}}
\newcommand{\Exponential}{\mathsf{Exponential}}
\newcommand{\dist}{d}
\newcommand{\algo}{\mathcal{A}}

\newcommand{\parens}[1]{\left( #1 \right)}

\newcommand{\polylog}{\mathrm{polylog}}
\newcommand{\Prob}[1]{\mathbf{Pr}\left[#1\right]}
\newcommand{\Expect}[1]{\mathbf{E}\left[#1\right]}

\newcommand{\logstar}{\log^{*}}
\newcommand{\ddim}{k}

\newcommand{\nbd}[2]{\mathcal{B}_{#1}(#2)}
\DeclareMathOperator{\vor}{Vor}
\DeclareMathOperator{\poly}{poly}
\DeclareMathOperator{\OPT}{\mathsf{OPT}}

\usepackage[capitalize, nameinlink]{cleveref}
\crefname{theorem}{Theorem}{Theorems}
\Crefname{lemma}{Lemma}{Lemmas}
\Crefname{figure}{Figure}{Figures}
\Crefname{claim}{Claim}{Claims}
\Crefname{observation}{Observation}{Observations}

\newtheorem{theorem}{Theorem}[section]

\newtheorem{lemma}[theorem]{Lemma}
\newtheorem{observation}[theorem]{Observation}

\newtheorem{proposition}[theorem]{Proposition}

\theoremstyle{definition}
\newtheorem{definition}[theorem]{Definition}

\usepackage[utf8]{inputenc}

\newcommand{\varsha}[1]{{\color{blue}{#1 --Varsha}}}

\newcommand{\tom}[1]{\todo[backgroundcolor=red!25]{Tom: #1}}
\newcommand{\yijun}[1]{\todo[backgroundcolor=blue!25]{Yijun: #1}}

\title{
Low-Distortion Clustering in Bounded Growth Graphs
}
\author{Yi-Jun Chang \\ National University Singapore \and Varsha Dani \\ Rochester Institute of Technology \and Thomas P. Hayes \\ University at Buffalo}
\date{}

\begin{document}

\maketitle

\begin{abstract}
The well-known clustering algorithm of Miller, Peng, and Xu~(SPAA 2013) is useful for many applications, including low-diameter decomposition and low-energy distributed algorithms.  One nice property of their clustering, shown in previous work by Chang, Dani, Hayes, and Pettie~(PODC 2020), is that distances in the cluster graph are rescaled versions of distances in the original graph, up to an $O(\log n)$ distortion factor and rounding issues. Minimizing this distortion factor is important for efficiency in computing the clustering, as well as in further applications, once the clustering has been constructed.

We prove that there exist graphs for which an $\Omega\left(\log^{1/3} n \right)$ distortion factor is necessary for any clustering.  We also consider a class of nice graphs which we call \emph{uniformly bounded independence} graphs.  These include, for
example, paths, lattice graphs, and ``dense'' unit disk graphs.
For these graphs, we prove that
clusterings of constant distortion always exist, and moreover, we give an efficient distributed algorithm to construct them.
Our clustering algorithm is based on Voronoi cells centered at the vertices of a maximal independent set in a suitable power graph.
   
Applications of our new clustering include low-energy simulation of distributed algorithms in the $\LOCAL$, $\CONGEST$, and $\RADIO$ models, as well as efficient approximate solutions to distributed combinatorial optimization problems.  We complement these results with matching or nearly matching lower bounds.
\end{abstract}

\thispagestyle{empty}

\newpage 
{ 
\tableofcontents
}
\thispagestyle{empty} %
\newpage


\setcounter{page}{1}

\section{Introduction}
Clustering, broadly speaking, refers to algorithms which partition the vertex set of a graph into
subsets which induce connected subgraphs.  In addition to being a fundamental object of study in graph
theory, clusterings have proven to be a useful tool in the design of graph algorithms, and also in
distributed computation.  Often, additonal properties are desired,
such as for the clusters to have small diameter, or for there to be few edges between clusters,
or for the cluster graph (also known as the quotient graph) to have certain properties, such as a bound on the degrees.

We consider the question of clusterings with low \emph{distance distortion}, defined formally
in \Cref{sec:clust-intro}.  Informally, we would like clusterings for which distances in the base graph
are approximately proportional to corresponding distances in the cluster graph; by \emph{scale}, we mean the constant of proportionality, and by \emph{distortion} we mean a 
particular quantification of ``how approximately.''  Such clusterings have been studied before.  Miller, Peng, and Xu~\cite{MPX}
defined a clustering algorithm which was shown by Chang, Dani, Hayes, and Pettie~\cite{Chang20bfs} to have $O(\log n)$
distance distortion.  They posed as an open question whether this is the best possible.

In the present work, we present both positive and negative results on this question.  On the one hand, we prove that there
exist graphs and distance scales for which every clustering has distortion $\Omega\left(\log^{1/3} n \right)$; in particular, the lower bound applies to any family of
regular graphs whose girth and diameter are both $\Theta(\log n)$.  On the other hand, we show that, for a fairly broad class of
``nice'' graphs, including some that are frequently studied as models for distributed networks, it is possible to attain $O(1)$ distortion
at all distance scales, and we give efficient distributed algorithms that achieve this.

\subsection{Network Models}
In the standard $\LOCAL$ model~\cite{Linial92} of distributed computing, a network is modeled as a graph $G=(V,E)$ so that each vertex $v \in V$ is a computer and each edge $e \in E$  is a bidirectional communication link. Each vertex $v \in V$ is equipped with a distinct $O(\log n)$-bit identifier $\ID(v)$. The communication proceeds in synchronous rounds: In each round, each vertex $v \in V$ receives the messages sent from its neighbors, performs some arbitrary local computation, and sends a message of arbitrary length to each of its neighbors. The $\CONGEST$ model is a variant of the $\LOCAL$ model where the size of each message is  $O(\log n)$ bits. 

We also consider the extension of $\CONGEST$ to the more challenging \emph{radio network} model. In the $\RADIO$ model, in each round of communication, each vertex can choose whether to transmit. In case a vertex chooses to transmit, it must transmit the same $O(\log n)$-bit message to all its neighbors. For each vertex $v$, it successfully receives a message from a neighbor $u$ if $u$ is the only neighbor of $v$ that transmits a message in this round.

Throughout the paper, we write $n = |V|$ and $m = |E|$. Let $d(u, v)$  denote the graphical distance between vertices $u$ and $v$ in $G$. We also write $d_G(u,v)$ when disambiguation is needed. Let $\nbd{u}{r}$ denote the ball of radius $r$ around $u$ in $G$, \emph{i.e.},  $\nbd{u}{r} = \{v \in V \, : \, d_G(u, v) \le r\}$.

\subsection{Graph Clustering}\label{sec:clust-intro}

In this paper, we focus on the distributed computation of graph clustering, where given an undirected graph $G=(V,E)$, the goal is to partition $V$ into subsets, each of which induces a connected subgraph of $G$, satisfying some desired requirements. Throughout the paper, for any $v \in V$, we write $[v]$ to denote the cluster that contains $v$ in the clustering under consideration.

\begin{definition}\label{def:clustering}
A \emph{clustering} of a graph $G$ is a partition of $G$ into connected subgraphs, called \emph{clusters}.
Each cluster is associated with a canonical representative vertex, known as its \emph{center}, as well as a
 BFS tree spanning the cluster and rooted at the center, where each vertex knows its level in the tree, its parent in the tree, and the cluster center.
\end{definition}


An equivalent way of specifying any clustering is to give a map $f : V \to V$, where, for every $v \in V$,
$f(v)$ equals the center of the cluster containing $v$.   We call this the \emph{cluster center}  map.

Graph clustering is a fundamental tool in distributed computing, with applications to
many different kinds of distributed algorithms. 
   In particular, clusterings that partition a graph into low-diameter components play a critical role in the complexity theory of local distributed graph problems~\cite{GhaffariKMU18}. There exist generic methods allowing us to obtain $(1\pm \epsilon)$-approximate solutions to an arbitrary covering or packing integer linear program problem in the $\LOCAL$ model~\cite{chang2023complexity,GhaffariKMU18}.   
   Moreover, it is known~\cite{GhaffariKMU18} that any \emph{sequential local} polylogarithmic-round algorithm can be converted into a   polylogarithmic-round distributed algorithm in the $\LOCAL$ model via such clustering.

\subparagraph{Voronoi Clustering.}
For a given set $S$ of \emph{centers}, arguably the most natural clustering is Voronoi clustering, also known as Voronoi diagram or
decomposition, or Dirichlet tesselation.  For concreteness, we use the unique $\ID$ associated with each vertex to break ties
in a canonical way.
\begin{definition}
For $\emptyset \ne S \subseteq V$, the Voronoi clustering centered at $S$, which we denote $\vor(S)$ is given by cluster center map
\[
f(v) = \argmin \{ \dist(s, v) \, \colon \, s \in S\}.
\]
Equivalently, each vertex joins the cluster of the center closest to it.  In case of a tie between several centers at
minimal distance from $v$, we take $f(v)$ to be the one with the smallest $\ID$.
\end{definition}

An important variation on this assigns additive weights to the various cluster centers.
\begin{definition}\label{def:AWVC}
Let $\emptyset \ne S \subseteq V$, and let $W: S \to \R$.  Then the associated
\emph{additively weighted Voronoi clustering} with weights $W$, which we denote $\vor(S,W)$, is given by cluster center  map
\[
f(v) =  \argmin \{ \dist(s,v) - W(s) \, \colon \, s \in S  \}.
\]
Any ties are broken in favor of the center with the smallest $\ID$, as before.
\end{definition}

We emphasize that, with additive weights, the actual number of clusters may be smaller than $|S|$, 
since some element $u \in S$ may prefer another center $v \in S$ over itself.
In this case, no vertex will select $u$ as the center, as $v$ is preferred compared with $u$ by the triangle inequality.


\subparagraph{Cluster Graph.} For a given clustering $\mathcal{C}$, 
we will frequently be interested in the associated \emph{cluster graph},
which is defined by the following quotient construction.
Let $V'$ be the set of clusters, and whenever an edge $\{v,w\}$ has its endpoints
in two distinct clusters, $[v]$ and $[w]$, then $E'$ contains edge $\{ [v], [w] \}$. 
In other words, the cluster graph is the quotient graph $G/\sim$ where $\sim$ is the 
    equivalence relation on vertices defined by the clustering $\mathcal{C}$.


\subsection{Crossing Edges and Distance Distortion}

Let $\mathcal{C}$ be a clustering, with associated cluster center map, $f$, 
and let $R \ge 1$.  We are particularly interested in the following two
metrics of how efficiently $\mathcal{C}$ represents graph $G$ with respect to a given \emph{scale factor} $R \geq 1$.  

\subparagraph{Few Crossing Edges.} The first metric is the fraction of crossing edges,
that is,
\[
\frac{|E'|}{|E|},
\]
where $E' = \{ \{v,w\} \in E \colon f(v) \ne f(w) \}$ is the set of \emph{crossing edges}. We would like most of the edges of the graph to have both endpoints within the same cluster.
    Ideally, the fraction of crossing edges should be $O(1/R)$, where $R$ is the scale factor.

\subparagraph{Low Distance Distortion.}  Our second goal is \emph{low distortion}, which is, informally, 
the extent to which distances in the cluster graph are close to being
distance in $G$, rescaled by the scale factor $R$. 
Intuitively, we would like the relationship
    $d \approx  R d'$ to hold, where $R \ge 1$ is the scale factor, $d$ represents the graphical distance in the original graph $G$, and $d'$ represents the graphical distance in the cluster graph $G'$.
Since graphical distances are
integers, there is  a limit to what  we can expect  here,  especially 
for nearby pairs of vertices.   We resolve this issue by regularizing using an additive term before
comparing them: $1 + (d/R) \approx 1 + d'$, which makes sense even when one or both of $d, d'$ equals zero. 
This motivates the first part of the following definition.
The second part of the definition requires that the \emph{strong diameter} of the clusters should be small, capturing the intuition that the scale factor $R$ is the ``intended'' cluster diameter bound.

\begin{definition} \label{def:distortion}
Let $(\mathcal{C}, R)$ be a clustering with scale factor $R$. 
Let $d$ denote graphical distance in $G$, and $d'$ denote graphical distance in the cluster graph.
We define the \emph{distance distortion} of $(\mathcal{C},R)$ as the minimum $C \ge 1$ to satisfy the following conditions.
\begin{enumerate}
    \item \label{cond1} For every $v,w \in V$, we have
    \[
\frac{1}{C} \le \frac{1 + (d(v,w)/R)}{1 + d'([v],[w])} \le C.
\]
    \item \label{cond2} The diameter of the subgraph induced by each cluster is at most 
    $C R$.  
\end{enumerate}
\end{definition}

 In \cref{def:distortion}, we emphasize that \cref{cond1,cond2} are logically independent. 
The \emph{weak diameter} of a vertex subset $S \subseteq V$ is defined as $\max_{u\in S, v\in S} \dist_G(u,v)$. While \cref{cond1} \emph{does} imply that the weak diameter of each cluster is $O(CR)$, \cref{cond1} \emph{does not} imply an $O(CR)$ cluster diameter upper bound. For example, let $G$ be the result of concatenating all the $n-1$ leaf vertices of an $n$-vertex star graph into a path $P$, and consider the clustering $\mathcal{C}$ where $P$ forms a cluster and the star center forms a single-vertex cluster. Let $R = 1$. The clustering $(\mathcal{C}, R)$ satisfies \cref{cond1} for $C = 3$, as $d(v,w) \in \{0, 1, 2\}$ and $d'([v],[w])\in \{0,1\}$ for all $v \in V$ and $w \in V$. This same clustering $(\mathcal{C}, R)$, however, does not satisfy \cref{cond2} for any $C < n-2$, as the diameter of $P$ is $n-2$.

\subsection{MPX Clustering}

Miller, Peng, and Xu~\cite{MPX} defined a particularly nice randomized algorithm which, for any input graph, 
produces a clustering that, with high probability, has 
both few crossing edges~\cite{MPX} and low distance distortion~\cite{Chang20bfs}.  
We will refer to this as ``the MPX clustering algorithm.''
Additionally, their algorithm can easily be implemented as a distributed algorithm, where the communication network 
is the graph to be decomposed.  

Formally, the MPX clustering algorithm computes an additively weighted Voronoi clustering, where
every vertex is a potential cluster center, and the weights are i.i.d.~exponentially distributed with mean $R$.
In other words, each vertex $v$ independently samples a weight $W(v)$ from the exponential distribution with mean $R$, and then each
vertex $u$ joins a cluster ``centered'' on $x$, where $x$ minimizes $d(u,x)-W(x)$. Using the terminology of \cref{def:AWVC}, for a given scale factor $R$, the MPX clustering is simply $\vor(S,W)$ with $S = V$ and $W(v) \sim \Exponential(1/R)$.

A useful interpretation of the MPX algorithm is as follows. Each vertex $v$ initiates a BFS tree rooted at $v$ at time $t= -W(v)$. The BFS trees grow at the rate of one hop per unit of time. If the first BFS tree to reache $v$ is rooted at $u$, then $v$ joins the cluster of $u$.



\begin{restatable}[\cite{MPX,Chang20bfs}]{theorem}{thmMPX}\label{thm:MPX}\label{prop:MPX-distortion-ub}
    For any graph $G = (V,E)$ and any parameter $R \ge 1$,
    the MPX clustering algorithm produces a clustering for which the
    expected number of crossing edges is $O(m/R)$ and 
    the distance distortion is $O(\log n)$ with probability $1 - 1/\poly(n)$.
\end{restatable}
\begin{proof}
    The bound on the number of crossing edges is from~\cite{MPX}. The bound on the distance distortion is from~\cite[Lemma 2.2]{Chang20bfs} by setting $R = 1/\beta$. 
\end{proof}




MPX clustering has a wide range of applications in  distributed computing: distributed approximation~\cite{BHKK16,chang2023complexity,Censor17maxcut,FFK21,faour2021approximating}, distributed property testing~\cite{Even17,levi2021property},  distributed spanner constructions~\cite{elkin2018efficient,forsterOPODIS2021}, distributed expander decomposition construction~\cite{chang2021near,ChangS20},
  and radio network algorithms~\cite{Chang20bfs,chang2018energy,CzumajD17,dani2022wake,haeupler2016faster}. 



\subsection{New Results}
We set out to answer one unsolved question, namely: Is it possible to improve on the construction of MPX, and in particular to reduce its distance distortion?  The importance of the question stems from the wide application of MPX. Notably, in the radio network algorithms by Chang, Dani, Hayes, and Pettie~\cite{Chang20bfs} and by Dani and Hayes~\cite{dani2022wake}, the distortion contributes significantly to their overall cost guarantees. Thus, improved distortion without sacrificing the other nice properties of MPX would immediately improve their results.

\subsubsection{Distortion Lower Bound}
The $O(\log n)$ distortion bound in \cref{thm:MPX} for analyzing the MPX clustering algorithm cannot be significantly improved, even for very nice graphs. For example, when $G$ is a cycle, with probability $1 - 1/\poly(n)$, the MPX clustering with $R = \Theta(\log n / \log \log n)$ contains a cluster with diameter $\Omega(R \log n)$ and a sequence of $\Omega(\log n / \log \log n)$ consecutive single-vertex clusters. This means that, although the fraction of crossing edges
is the optimal $\Theta(1/R)$, distances are somewhat distorted due to the non-uniformity of the clusters: At scale $R = \Theta(\log n / \log \log n)$, the distances can be \emph{overestimated} by a distortion factor of $\Omega(\log n)$ and \emph{underestimated} by a distortion factor of $\Omega(\log n / \log \log n)$. See \cref{sec:MPXanalysis} for details.



This naturally leads us to the question as to whether some other algorithm could find clusterings where the distortion is $O(1)$, or whether there exist graphs for which polylogarithmic distortion is unavoidable.
In \cref{sec:dist-lb}, we answer this question, negatively.

\begin{restatable} {theorem}{thmDistLB}\label{thm:distortion-lb-summarized}\label{cor:distortion-lb}
    For every constant $d\ge 3$, there exist $d$-regular graphs whose girth and diameter are both $\Theta(\log n)$ such that for all $R = O\left(\log^{1/3} n\right)$, every clustering with scale factor $R$ distorts distances by $\Omega(R)$.
\end{restatable}




Indeed, \cref{thm:distortion-lb-summarized}  establishes that there are graphs and scale factors for which no clustering can have a distortion of $o\left(\log^{1/3} n\right)$, so now we know that polylogarithmic distortion is unavoidable. It remains an intriguing open question to determine the right exponent in $[1/3, 1]$ for the optimal distortion bound.

\subsubsection{New Clustering}
While for general graphs we came up with a negative answer, one could still hope that for some ``nice'' classes of graphs, it is possible to construct clusterings with constant distortion. For example, to obtain such clusterings for a cycle, it suffices to partition the cycle into paths of length $\Theta(R)$. This leads us to explore special graph classes for which such clusterings are possible.

Our main contribution in this paper is a \emph{simple} new clustering algorithm that completely meets \emph{all} of our objectives: few crossing edges, constant distortion, and efficient distributed constructions, for a fairly big class of graphs.
 
\begin{definition} 
A graph $G = (V,E)$ is of {\bf bounded independence} if there are constants $\gamma, \ddim >0$ such that for all $r > 0$ and $v \in V$,
every independent set $S \subseteq \nbd{r}{v}$ satisfies $|S| \le \gamma r^{\ddim}$.
\end{definition}

Bounded independence graphs, also known as bounded growth graphs in some works, have been extensively studied in the distributed setting~\cite{kuhn2005local,kuhn2005locality,gfeller2007randomized,schneider2010optimal}. Notably, Schneider and Wattenhofer~\cite{schneider2010optimal} presented a highly efficient MIS algorithm for such graphs in the $\CONGEST$ model, which will play a crucial role in our work. A primary motivation for investigating bounded independence graphs is that they cover \emph{unit disk graphs}, which are commonly employed in modeling wireless networks for the study of distributed graph algorithms~\cite{alzoubi2002message,burkhart2004does,funke2006simple,gao2001geometric,kuhn2003worst}.

\begin{definition}
    Fix $R>0$. We write $G^{\le R}$ to denote the graph with the same vertex set as $G$, but where there is an edge between $u$ and $v$ if $d_G(u,v) \le R$. We call $G^{\le R}$ the {\bf power graph} of $G$ with respect to parameter $R$.
\end{definition}

If $G$ is a graph with bounded independence and $R$ is a fixed \emph{constant}, then $G^{\le R}$ also has bounded independence. However, the actual independent set size bound gets much worse in the power graph. Indeed, if the polynomial bounding the size of independent subsets of $r$-balls in $G$ is  $\gamma r^{\ddim}$, then the corresponding bound for $G^{\le R}$ is $\gamma_R r^{\ddim}$, where $\gamma_R = \gamma R^{\ddim}$, as $r$-balls in $G^{\le R}$ are $rR$-balls in $G$. Since $\gamma_R$ depends on the parameter $R$, this observation is useful only when $R$ is small.

We consider graphs whose power graphs $G^{\le R}$ are of bounded independence  \emph{uniformly} for all $R$.

\begin{definition}  
A graph $G = (V,E)$ is of {\bf \emph{uniformly} bounded independence} if there exist constants $\gamma, \ddim >0 $ such that for all $R >0$, $G^{\le R}$ is  bounded independence with parameters $\gamma, \ddim$.
\end{definition}

On the face of it, uniformly bounded independence seems like a very strong property and one might worry that it is impossible to satisfy. In~\cref{appsec:bdi-eg}, we show that various examples of interesting graph classes do exhibit this property, including paths, cycles, grids, lattices, and random geometric graphs, among others. 
Random geometric graphs~\cite{gilbert1961random} are widely used to model wireless sensor networks~\cite{haenggi2009stochastic}. Distributed algorithms on grids have been extensively studied~\cite{balliu2022local,bonamy_et_al:LIPIcs.DISC.2018.12,brandt2017lcl,grebik2023local}. For these graphs, we show in \cref{sec:low-dist-clust} that clusterings with constant distortion are possible, and moreover, can be efficiently constructed by distributed algorithms. 


\begin{restatable}{theorem}{thmDistCross}\label{thm:distort-and-crossing}
    Let $G$ be a graph and $R \ge 1$ be a scale factor such that $G^{\le R}$ has bounded independence. There exists a clustering of $G$ with a constant distance distortion
    at scale $R$ and an $O(1/R)$ fraction of crossing edges.  
\end{restatable} 

If $G$ is of uniformly bounded independence, then the condition of 
\cref{thm:distort-and-crossing} is satisfied for any choice of the scale factor $R$. As we will later see, an application of low distortion clusterings is in low energy simulation of distributed algorithms. To achieve this we would like to follow the multi-scale clustering framework of Dani and Hayes~\cite{dani2022wake}. In this framework, it is necessary to simultaneously build low-distortion clusterings at multiple scales that somehow interact well with each other.
Thus, we are interested in graphs for which bounded independence carries through to the power graphs \emph{uniformly at all scales.}  

 Our construction of the clustering of \cref{thm:distort-and-crossing} is completely different from MPX: Instead of using randomized start times to determine cluster centers, which leads to large variations in cluster diameters, we first find a maximal independent set (MIS) in the power graph $G^{\le R}$, and use those as cluster centers to partition $G$ into Voronoi cells, which leads to far more uniform cluster diameters.  
Specifically, cluster diameters for MPX, even on really nice graphs like the square grid, range from $1$ to $R \log n$, but cluster diameters for our construction are always $\Theta(R)$.  The bounded independence property can be shown to imply low distortion for this clustering. To ensure that only $O(m/R)$ edges are crossing, we tweak this construction slightly using additive weights.

Although the above idea is almost as simple as the MPX algorithm, implementation is more difficult because we need to solve MIS to find the cluster centers. In the end, it can still be implemented efficiently thanks to excellent prior work on distributed MIS. As we will discuss, existing techniques for constructing maximal independent sets can be leveraged to convert the algorithm of \cref{thm:distort-and-crossing} into an efficient distributed algorithm, specifically
in the $\LOCAL$ model, which is essential for most of our algorithmic applications.
For $\CONGEST$ and $\RADIO$, further challenges are caused by the fact that we need a MIS for the power graph, $G^{\leq R}$, but the algorithm needs to run on the base graph $G$.

Given that MPX is already widely applied across various domains, our algorithm yields improved results for many problems for a large and practically useful class of graphs. The improvement from $O(\log n)$ to $O(1)$ is significant for problems whose complexity is already small. As we will later see, \cref{thm:distort-and-crossing} allows us to obtain tight round complexity bounds and nearly tight energy complexity for many natural and well-studied distributed problems, such as $(1-\epsilon)$ approximation for the maximum independent set and maximum matching.




\subsection{Distributed Computing Applications}
In \cref{sec:low-dist-clust}, we show that, for graphs $G$ such that 
$G^{\le R}$ is of bounded independence, very good low-distortion clusterings can be constructed efficiently.

\begin{restatable}{theorem}{thmLOCALclustering}\label{thm:LOCAL-clustering}
Let $R \ge 1$.
There is a randomized $\LOCAL$ algorithm which, when run on a graph $G$ such that 
$G^{\le R}$ is of bounded independence,
produces a clustering at scale $R$ that satisfies the following conditions in $O(R \logstar n)$ rounds.
\begin{itemize}
    \item The distance distortion is constant, so the diameter of each cluster is $O(R)$.
    \item Each edge is a crossing edge with probability $O(1/R)$, so the clustering has an $O(1/R)$ fraction of crossing edges in expectation.
\end{itemize}
\end{restatable}

Moreover, we show that the algorithm of \Cref{thm:LOCAL-clustering} can be \emph{derandomized} using a proper distance coloring with $O(1)$ colors and the method of conditional expectations.
\begin{restatable}{theorem}{thmLOCALclusteringDet}\label{thm:LOCAL-clustering-det}
Let $R \ge 1$.
There is a deterministic $\LOCAL$ algorithm which, when run on a graph $G$ such that 
$G^{\le R}$ is of bounded independence,
produces a clustering at scale $R$ that satisfies the following conditions in $O(R \logstar n)$ rounds.
\begin{itemize}
    \item The distance distortion is constant, so the diameter of each cluster is $O(R)$.
    \item The clustering has an $O(1/R)$ fraction of crossing edges.
\end{itemize}
\end{restatable}

As we will later see, both the fraction of crossing edges $O(1/R)$ and the round complexity of the construction algorithm $O(R \logstar n)$ are asymptotically \emph{optimal}.

Applications of our new clustering algorithms include low-energy simulation of distributed algorithms 
and efficient approximate solutions to distributed combinatorial optimization problems. We also complement these results with matching or nearly matching lower bounds. 

While \Cref{thm:LOCAL-clustering-det} seems to be strictly better than \Cref{thm:LOCAL-clustering}, as we will later see, some of our applications specifically require the property of \Cref{thm:LOCAL-clustering} that each edge is a crossing edge with probability $O(1/R)$, so \Cref{thm:LOCAL-clustering,thm:LOCAL-clustering-det} are incomparable.

\subsubsection{Low-Energy Simulation}
In standard distributed computing one is often interested in round complexity, \emph{i.e.}, the number of rounds taken by an algorithm, and message complexity, \emph{i.e.}, the number of messages sent by vertices during the algorithm. While these measures are good descriptors of the complexity of a distributed algorithm, they do not account for the portion of the running time when a vertex is simply waiting for information from its neighbors. \emph{Energy complexity}---also known as \emph{awake complexity}---captures the idea that a vertex should not be charged for rounds when it is \emph{predictably} not participating in the action of the algorithm. The predictability here is key: If a vertex knows it will not send or receive messages in certain rounds, then it can power down for those rounds of the algorithm. On the other hand, if it is powered down when it should be receiving a message, then the message is lost and could compromise the correctness of the algorithm. The goal, then, is to design distributed algorithms that enable vertices to have scheduled downtimes when they are far from the action in the algorithm, thereby saving energy.

We model energy usage as follows. In each round of the algorithm, a vertex $v$ can choose to be asleep or awake. Every round where $v$ chooses to be awake costs one unit of energy for $v$. A vertex can only participate in the algorithm, \emph{i.e.}, communicate with other vertices, while it is awake. Notably, in the $\LOCAL$ model, an awake vertex may exchange complete information about its state with all of its awake neighbors. The setting where vertices are allowed to sleep is known as the \emph{sleeping} model~\cite{ChatterjeeGP20}. The energy usage of a vertex is its number of awake rounds. The energy complexity of an algorithm is the maximum energy usage of a vertex $v$ among all $v \in V$.


Designing low-energy distributed algorithms for specific problems has been the subject of many prior works. The earlier research on this topic mostly focused on single-hop radio networks~\cite{BenderKPY16,ChangKPWZ17,JurdzinskiKZ02podc,lavault2007quasi,nakano2000randomized}. Recent works extended this line of research to multi-hop radio networks~\cite{chang2018energy,Chang20bfs,chang2023energy,dani2022wake,dani2021wake} as well as the $\LOCAL$ and $\CONGEST$ models~\cite{augustine2022distributed,barenboim2021,ChatterjeeGP20,dani2022wake,dufoulon2023distributed,ghaffari2022average,ghaffari2023distributed}.

Low-distortion clusterings have been used in the design of energy-efficient algorithms for breadth-first search~\cite{Chang20bfs} and generic algorithm simulation~\cite{dani2022wake} in the radio network model. These applications critically rely on the low-distortion property of the clustering.
It is not enough to know that good clusterings exist; we also need efficient 
algorithms to find them.  

Our first application is to show that any $\LOCAL$ algorithm on a graph of {uniformly} bounded independence having round complexity $O(t \logstar n)$
can have its energy complexity reduced to $O(\log t  \logstar n)$ without affecting its asymptotic round complexity. Informally, this implies that for \emph{any} problem within this graph class, the energy complexity can be \emph{exponentially} improved compared to the round complexity essentially for free!

\begin{restatable}{theorem}{thmLOCALsim}\label{thm:localsim}
Let $G= (V,E)$ be a graph with uniformly bounded independence. Let $t \geq 1$. Given any $t \log^\ast n$-round $\LOCAL$ algorithm $\mathcal{A}$, there is an algorithm $\mathcal{A'}$ that computes the same function as $\mathcal{A}$, in $O(t \log^\ast n)$ rounds and using $O(\log t \logstar n)$  energy per vertex.
\end{restatable}

While \cref{thm:LOCAL-clustering-det,thm:LOCAL-clustering} only requires constant distortion at scale $R$, \cref{thm:localsim} requires constant distortion uniformly at all scales because the algorithm is based on constructing low-distortion clusterings at \emph{multiple scales}.
The main proof idea behind \cref{thm:localsim} is similar to the hierarchical clustering in~\cite{Chang20bfs,dani2022wake}. The idea is that within a cluster at a particular level, vertices know how far they are from the root of the cluster. This means that within the cluster, they can send and receive messages on a schedule, and sleep the rest of the time. Moreover, inter-cluster messages can also be arranged to be sent at predictable times. Therefore, a clustering at scale $2^{i+1}$ can be built energy-efficiently once we have a clustering at scale $2^i$. 

To achieve the generic low-energy simulation of \cref{thm:localsim}, we construct a low-congestion overlapping clustering where $\mathcal{B}_t(v)$ is entirely contained in some cluster for every vertex $v \in V$. Leveraging the unlimited bandwidth of the $\LOCAL$ model, the BFS trees within each cluster enable an efficient low-energy simulation of any $t$-round $\LOCAL$ algorithm. Such clustering can be obtained through slight modifications to our construction of the clustering of \cref{thm:distort-and-crossing} at scale $t$.


\subparagraph{Extension to $\CONGEST$ and $\RADIO$.}
Dani and Hayes~\cite{dani2022wake} showed that, up to polylogarithmic factors in the 
round and energy complexity, 
there is a ``best'' generic method for reducing the
energy cost of any algorithm in the $\RADIO$ model.
Building on their work, in \cref{sec:newCONGESTstuff}, we obtain an \emph{improved} result for graphs with uniformly bounded independence by replacing their use of MPX clustering with our less distance-distorting clustering.  
The main challenge to realize this idea is to build the required multi-scale clustering in $\CONGEST$ and $\RADIO$.

A key idea to address this challenge is the following modification to the multi-scale clustering algorithm: Instead of building the next-level clustering from scratch, we choose the new set of cluster centers to be a maximal subset of the existing cluster centers. While they do not form a maximal independent set in the considered power graph, we show that they are good enough for our purpose.

\subsubsection{Combinatorial Optimization}
Our clustering constructions can be used to give new fast approximation algorithms for certain 
combinatorial optimization problems.

It is well-known~\cite{BHKK16,chang2023complexity,Censor17maxcut,FFK21,faour2021approximating} that decomposing the communication network into low-diameter clusters with few crossing edges is useful for efficiently approximately solving various optimization problems. The high-level proof idea is to drop all edges between different clusters and solve the 
optimization problem on the resulting graph. Since each edge is a crossing edge with probability $O(1/R)$ in \cref{thm:localsim}, a $1 - O(1/R)$ fraction
of a fixed optimal solution survives the deletion of crossing edges, so the elimination of these
edges does not change the optimum by much. After the deletion of crossing edges, the problem is now easy to solve because the graph is now decomposed into components of small diameter. 
Examples of combinatorial optimization problems that are amenable to this approach include {maximum matching} and {maximum cut}.

Moreover, such randomized algorithms can be derandomized using the \emph{method of conditional expectations} to obtain deterministic algorithms attaining the same guarantees. The use of the method of conditional expectations in distributed graph algorithms was pioneered in~\cite{censor2020derandomizing} and subsequently used in many works~\cite{bamberger2020efficient,deurer2019deterministic,faour2023local,ghaffari2018derandomizing,ghaffari2022deterministic,ghaffari2023faster}. The high-level idea is to set the values of the random variables sequentially to optimize the conditional expectation. To efficiently implement this approach as a distributed algorithm, we demonstrate that the value of many random variables can be fixed simultaneously in parallel.

\begin{restatable}{theorem}{thmApxOpt}\label{thm:thmApxOpt}
   The following holds for these two problems: {\bf Maximum Matching} and {\bf Maximum Cut}.
    Let $G$ be a graph with uniformly bounded independence, and let $0 < \epsilon< 1$.  There is a deterministic $\LOCAL$ algorithm that
    outputs a feasible solution whose size is at least $(1 - \epsilon)\OPT$, where $\OPT$ is the value of an optimal solution, with round complexity $O\left(\frac{1}{\epsilon} \log^* n\right)$ and energy complexity $O\left(\log\frac{1}{\epsilon} \log^* n\right)$.
\end{restatable}

The proof of \cref{thm:thmApxOpt} is given in \cref{sec:approx}.


\subsection{Bounded Growth Graphs}  

In this work, we also revisit the MPX clustering algorithm~\cite{MPX} and present a variant of it that achieves better guarantees in \emph{bounded growth graphs}.

\begin{definition}\label{def:bounded_growth} 
A graph $G = (V,E)$ is of {\bf bounded growth} if there are constants $\beta, \ddim >0$ such that for all $v \in V$, for  all $r$,  $|\nbd{r}{v}| \le \beta r^{\ddim}$.
\end{definition} 

A bounded growth graph is automatically bounded independence, so bounded growth is a more stringent requirement than bounded independence. 
Cycles, grids, and lattices in dimension $d$ are natural examples of bounded growth graphs. 
Distributed algorithms on bounded growth graphs have been studied in~\cite{abraham2005name,bernshteyn2023borel,schneider2011distributed}.



In \cref{sect:mpx}, we show that in bounded growth graphs, a variant of the MPX algorithm can be used to construct clusterings with smaller maximum cluster diameters than the usual MPX clusters.

\begin{restatable}{theorem}{thmMPXboundedgrowth}\label{thm-mpx-bounded-growth}
Given a parameter $R$, for any bounded growth graph $G=(V,E)$, there is an $O(R \log R)$-round algorithm in the $\CONGEST$ model that decomposes the vertex set $V$ into clusters meeting the following properties.
\begin{itemize}
    \item The subgraph induced by each cluster has diameter $O(R \log R)$.
    \item Each edge $e \in E$ is a crossing edge with probability $O(1/R)$.
\end{itemize}
\end{restatable}

\Cref{thm-mpx-bounded-growth} is an improvement over the original MPX clustering, whose cluster diameter bound is $O(R \log n)$, as the maximum value of the additive weight $W(v)$ among all vertices $v \in V$ is most likely  $\Theta(R \log n)$. Intuitively, the $O(\log n)$ factor is due to the \emph{union bound} over $n$ vertices in $V$.

The key idea behind our proof of \cref{thm-mpx-bounded-growth} is to leverage the bounded growth property: Within any ball of radius $\Theta(R \log R)$, the number of vertices is polynomial in $R$, implying that the maximum value of the additive weight $W(v)$ among all vertices in the ball is likely to be $\Theta(R \log R)$. Intuitively, this implies that imposing a \emph{cutoff} of $\Theta(R \log R)$ to cap the maximum value of all additive weights should have minimal impact on the analysis of the MPX algorithm on bounded growth graphs. Implementing the cutoff allows us to obtain a better cluster diameter bound.

As a consequence of \Cref{thm-mpx-bounded-growth}, we obtain improved results for approximating combinatorial optimization problems in bounded growth graphs.

\begin{restatable}{corollary}{thmMPXapx}\label{cor-mpx-bounded-growth}
   The following holds for these two problems: {\bf Maximum Matching} and {\bf Maximum Independent Set}.
    Let $G$ be a graph with bounded growth, and let $\epsilon = n^{-o(1)}$.  There is a randomized $\LOCAL$ algorithm that
    outputs a feasible solution whose size is at least $(1 - \epsilon)\OPT$ with probability $1 - 1/\poly(n)$, where $\OPT$ is the value of an optimal solution, with round complexity $O\left(\frac{1}{\epsilon} \cdot \log \frac{1}{\epsilon} \right)$.
\end{restatable}

\subsection{Lower Bounds}

We complement our upper bound results with matching or nearly matching lower bounds. The following lower bound was shown by Lenzen and  Wattenhofer~\cite{LenzenW08}.

\begin{theorem}[\cite{LenzenW08}]\label{thm-lb-time}
 Any deterministic algorithm that solves the $(1-\epsilon)$-approximate maximum independent set problem requires $\Omega\left(\frac{\log^\ast n}{\epsilon}\right)$ rounds in cycles in the $\LOCAL$ model. 
\end{theorem}

\cref{thm-lb-energy} immediately implies that the round complexity $O\left(\frac{\log^\ast n}{\epsilon}\right)$ in \cref{thm:thmApxOpt} is asymptotically optimal.

Moreover, \cref{thm-lb-time} implies that any deterministic algorithm that removes at most $\epsilon$ fraction of the edges to decompose a cycle into connected components of diameter $O(1/\epsilon)$ must require $\Omega\left(\frac{\log^\ast n}{\epsilon}\right)$ rounds in the $\LOCAL$ model, as such a clustering allows us to construct a $(1-O(\epsilon))$-approximate maximum independent set of a cycle using an additional $O(1/\epsilon)$ rounds, as we explain below. Let each $O(1/\epsilon)$-diameter cluster locally compute a maximum independent set of this cluster, and then let $S$ be the union of these independent sets. The size of $S$ is at least $\lfloor n/2 \rfloor$. For each edge $e=\{u,v\}$ that crosses two clusters such that both $u$ and $v$ are in $S$, remove any one of them from $S$. The result is an independent set of size at least $\lfloor n/2 \rfloor - \epsilon n = (1 - O(\epsilon)) \cdot \OPT$, where $\OPT = \lfloor n/2 \rfloor$ is the size of a maximum independent set of a cycle. 

As cycles are uniformly bounded independence graphs, by a change of variable $R = \Theta(1/\epsilon)$, the above argument implies that the round complexity $O(R \log^\ast n)$ of our deterministic clustering algorithm of \cref{thm:LOCAL-clustering-det} is asymptotically \emph{optimal} in uniformly bounded independence graphs. The reason is that the fraction of crossing edges of any clustering of a cycle at scale $R$ with constant distortion must be $O(1/R)$, since otherwise there is a path of length $R$ intersecting $\omega(1)$ clusters, meaning that the distance distortion for the two endpoints of the path is $\omega(1)$.

In \cref{sect:lb}, we prove the following two theorems.
 
 \begin{restatable}{theorem}{thmLBenergy}\label{thm-lb-energy}
 Any algorithm that solves the $(1-\epsilon)$-approximate maximum independent set problem with probability at least $0.99$ requires $\Omega\left(\log \frac{1}{\epsilon}\right)$ energy in cycles in the $\LOCAL$ model.
\end{restatable}

\cref{thm-lb-energy} immediately implies that  the energy complexity $O\left(\log\frac{1}{\epsilon} \log^* n\right)$ in \cref{thm:thmApxOpt} is nearly optimal, up to a tiny $O(\log^\ast n)$ factor. 

Moreover, \cref{thm-lb-energy} implies an $\Omega(\log t)$ lower bound for the low-energy simulation in \cref{thm:localsim}. If the energy complexity of \cref{thm:localsim} can be improved to $o(\log t)$, then the improved low-energy simulation, combined with the approximate maximum independent set algorithm of \cref{cor-mpx-bounded-growth}, implies that a $(1-\epsilon)$-approximate maximum independent set of a cycle can be computed using $o\left(\log\frac{1}{\epsilon}\right)$ energy in $\LOCAL$, contradicting the lower bound of \cref{thm-lb-energy}.


\begin{restatable}{theorem}{thmLBtime}\label{thm-lb-time-bounded-growth}
 Any algorithm that solves the $(1-\epsilon)$-approximate maximum independent set problem  in expectation requires $\Omega\left(\frac{1}{\epsilon} \cdot \log \frac{1}{\epsilon} \right)$ rounds in bounded growth graphs in the $\LOCAL$ model. 
\end{restatable}

\cref{thm-lb-time-bounded-growth} implies that the round complexity of $O(R \log R)$ in \cref{cor-mpx-bounded-growth} cannot be further improved. Consequently, our $O(R \log R)$-round variant of the MPX algorithm in bounded growth graphs stated in \cref{thm-mpx-bounded-growth} attains the asymptotically optimal round complexity.


\subsection{Related Work}

Kuhn, Moscibroda, Nieberg, and Wattenhofer~\cite{kuhn2005local} designed distributed approximation algorithms for bounded independence graphs.
Kuhn, Moscibroda, and Wattenhofer~\cite{kuhn2005locality} showed that a maximal independent set can be computed in $O(\log^\ast n)$ rounds in a class of bounded independence graphs with geometric information in the $\LOCAL$ model. The reliance on geometric information was removed in subsequent works~\cite{gfeller2007randomized,schneider2010optimal}. Schneider and Wattenhofer~\cite{schneider2010optimal} showed that for any bounded independence graphs, a maximal independent set can be computed in $O(\log^\ast n)$ rounds deterministically in the $\CONGEST$ model.
 
Bounded independence graphs arise naturally in geometric-based network models. In particular,  bounded independence graphs have been studied in the signal-to-interference-plus-noise-ratio (SINR) wireless network model~\cite{halldorsson2015local} and the radio network model~\cite{davies2023uniting,schneider2010use}.

The \emph{neighborhood independence} of a graph is the maximum number of independent neighbors of a single vertex in the graph.
Bounded neighborhood independence graphs have been studied in various distributed, parallel, and sublinear settings~\cite{assadi2019algorithms,barenboim2011distributed,fischer2017deterministic,milenkovic2020unified}.

  

\subsection{Roadmap}
In \cref{sec:dist-lb}, we show that there are graphs for which clusterings with constant distortion are not possible.
In \cref{sec:low-dist-clust}, we show how to construct low-distortion clusterings for graphs with uniformly bounded independence. 
In \cref{sec:energy}, we show how to use low-distortion clusterings to do a generic low-energy simulation in the $\LOCAL$ model. 
In \cref{sec:newCONGESTstuff}, we extend the low-energy simulation to $\CONGEST$ and radio networks.
In \cref{sec:approx}, we show how to use our new clusterings to get approximate solutions to combinatorial optimization problems. In \cref{sect:mpx}, we show that in bounded growth graphs, a variant of the MPX algorithm can be used to build clusterings with smaller maximum diameters than the usual MPX clusters, also resulting in improved results for approximating combinatorial optimization problems. In \cref{sect:lb}, we complement our algorithmic results with matching or nearly matching lower bounds.
In \cref{sec:MPXanalysis}, we show that the known $O(\log n)$ distortion bound for MPX clustering cannot be significantly improved, so a new approach is needed to lower the distortion.
In \cref{appsec:bdi-eg}, we show that there are interesting classes of graphs that have the bounded independence property.



\section{Distortion Lower Bounds}
\label{sec:dist-lb}

Chang, Dani, Hayes and Pettie~\cite{Chang20bfs} established that
the MPX clustering has logarithmic distance distortion, with high probability,
regardless of the choice of graph $G$ and the scale factor $R \ge 1$ (\cref{prop:MPX-distortion-ub}).
It is natural to ask whether, for all graphs and all distance scales $R \ge 1$, 
there always exists a clustering with $O(1)$ distortion.
We answer this question in the negative, by showing that
for any regular graph
with logarithmic girth, a lower bound of $\Omega\left(\log^{1/3} n\right)$ holds for some choices of $R$.

We begin with the observation that, for every graph, there are two trivial clusterings: one in which every cluster has size 1, and one
in which there is a single cluster of size $n$.  If we use one of these as our clustering for scale $1 \le R \le D$, where $D$ is the diameter of the graph, 
it is easy to check that 
the distance distortion is $\Theta(R)$ or $\Theta(D/R)$, respectively. Therefore, to get clusterings at all scales with distortion $O(\sqrt{D})$, even just the two trivial clusterings (the coarsest and finest possible partitions) suffice.

Surprisingly, we cannot do more than polynomially better than the trivial distortion bound $O(\sqrt{D})$ for general graphs! More concretely, our next result implies that, for $d$-regular graphs where the girth and diameter are both $\Theta(\log n)$, for each scale $R$, the distortion bound given by the best of these two trivial clusterings is within a polynomial of the optimal bound. 
Our result means that the $O(\log n)$ distortion upper bound for MPX clustering cannot be improved to better than 
$O\left(\log^{1/3} n\right)$ on general graphs, for any clustering scheme.

\begin{theorem} \label{thm:distortion-lb}
Let $G$ be a $d$-regular graph with girth $g$, and let $(\mathcal{C},R)$ be a clustering with scale factor $R$.
Then the distance distortion is $\Omega\left(\min\{R, \sqrt{g/R} \} \right)$.
\end{theorem}

\begin{proof} 
Suppose the distortion is smaller than claimed.
Let $D_0$ denote the maximum diameter of a cluster.
Then $D_0 \le \min\{R^2, \sqrt{gR}\}$, or else there is nothing to prove, as
two maximally distant points in the same cluster would give a contradiction. 

We may further assume that $D_0 \geq 1$ and $R = o(g)$. 
\begin{itemize}
    \item If $D_0 = 0$, then the distance distortion is already $\Theta(R)$.
    \item If $R = \Omega(g)$, then $\min\{R, \sqrt{g/R} \} = O(1)$, so there is also nothing to prove.
\end{itemize}
 Observe that $R = o(g)$ implies that $D_0 \le \min\{R^2, \sqrt{gR}\} < g/2$. That is, the diameter of every cluster is less than half the girth. 

Now, choose $\ell = \lfloor g/(D_0+1) \rfloor$, which ensures that $\ell(D_0+1) \le g$.  We have
$\ell = \Omega(\sqrt{g/R})$ because $D_0 \leq \sqrt{gR}$.  It will suffice to find a path of length $\ell$ in $G$
whose length in the cluster graph remains $\Theta(\ell)$. That is, $d'([s],[t])=\Theta(\ell)$, where $d'$ is the distance in the cluster graph and $s$ and $t$ are the two endpoints of the path. We claim that the existence of such a path implies that the distance distortion is $\Omega\left(\min\{R, \sqrt{g/R} \} \right)$.

\begin{itemize}
    \item Consider the case of $g \ge R^3$. We have $\ell = \Omega(R)$. The distance distortion is $\Omega(R)$ because \[\frac{1 + (d(s,t)/R)}{1 + d'([s],[t])} = \frac{\Theta(\ell/R)}{\Theta(\ell)}= \Theta(1/R).\]
    \item Consider the remaining case of $g < R^3$. We have $\ell = \Omega(\sqrt{g/R})$. The distance
distortion is  $\Omega(\min\{R,\ell\}) = \Omega\left(\min\{R, \sqrt{g/R} \} \right)$ because \[\frac{1 + (d(s,t)/R)}{1 + d'([s],[t])} = \frac{\Theta(1 +(\ell/R))}{\Theta(\ell)}= \Theta(1/\min\{R,\ell\}).\] 
\end{itemize}
 In the calculation above, we rely on the fact that any path of length $\ell$ is a shortest path in $G$, since $\ell = \lfloor g/(D_0+1) \rfloor \leq g/2$, where the inequality follows from the assumption that $D_0 \geq 1$.
 
Now, since all clusters have diameter at most $D_0$, which is less than half the girth,
it follows that each cluster is a \emph{tree}.  Since we additionally know $G$ is $d$-regular, it
follows by an easy induction that that the number of edges crossing out of any cluster is exactly
$d  + (d-2)i$, where $i$ is the number of internal edges (edges having both endpoints in the cluster).
Since each crossing edge is incident to two clusters, 
while each internal edge is incident to one, 
it follows that the fraction of edges in $G$ that are crossing is more than $(1- 2/d)$.

Fix a length $\ell <g$ and an edge $e=\{u,v\}$. There are $(d-1)^{a-1}$ paths that reach $u$ in exactly $a-1$ steps without passing through $e$, and $(d-1)^{\ell-a}$ paths of length $\ell-a$ that emanates from $v$ without passing through $e$. 

Since $\ell$ is less than the girth, the set of starting points and ending points must be disjoint (otherwise a cycle of length $\le \ell$ would have been formed.)  So, multiplying these, there are  $(d-1)^{\ell -1}$ paths with $e$ in the $a$th position. Summing over all possible positions, 
the number of paths of length $\ell$ containing a particular edge $e$ is exactly $\ell (d-1)^{\ell -1}$. 
In particular,  every edge is in the same number of paths of length $\ell <g$.  
Therefore, by averaging, there must exist a path of length 
$\ell$ for which more than a $(1 -  2/d)$ fraction of its edges cross between clusters.  
Let $u,v$  be the endpoints of such a path $P$.
Since the girth is $g$, every other path  from  $u$  to $v$ has length at least $g - \ell > \ell$, so $P$ is a shortest path and $d(u,v)=\ell$.  Moreover, for any $w, x$ on $P$, if $w$ and $x$ have a crossing edge anywhere between them, then they must be in different clusters (since we already noted that the shortest paths between vertices in the same cluster must remain within the cluster.) Thus $P$ intersects at least $(1-2/d)\ell +2$ distinct clusters and projects to a path of length greater than $(1-2/d)\ell$ between the clusters of $u$ and $v$ in the cluster graph. Now suppose that the cluster-graph-distance between the clusters of $u$ and $v$ is $q$. Then there is a path $P'$ from $u$ to $v$ of length at most $q D_0$. If $P'$ is different from $P$ then it must have length at least $g-\ell$, so that $\ell + q D_0 \ge g$.  But by our choice of $\ell$, we have $\ell(D_0+1) \le g$, which implies $q \ge \ell$.
Hence the
distance  from  $u$  to $v$ in the cluster  graph is at least $(1-2/d)\ell = \Theta(\ell)$, as desired. 
This completes the proof.
\end{proof}

Combining \cref{thm:distortion-lb} with a known fact in graph theory, we obtain the following result.

\thmDistLB*


\begin{proof}
    For every constant $d\ge 3$,  $d$-regular graphs whose girth and diameter are both $\Theta(\log n)$ are known to exist~\cite{Bollobas78}.  For example, when $d-1$ is prime and congruent to 1 mod 4, the Ramunujan graphs constructed by Lubotzky,  Phillips, and Sarnak~\cite{lubotzky1988ramanujan} have this property.
   For all $R = O\left(\log^{1/3} n\right)$, \cref{thm:distortion-lb} implies that the distance distortion of every clustering with scale factor $R$ is  $\Omega\left(\min\{R, \sqrt{g/R} \} \right) = \Omega(R)$, as $\sqrt{g/R} = \Omega(R)$ when $R = O\left(\log^{1/3} n\right)$.
\end{proof}

In light of \cref{cor:distortion-lb}, we  see that, if we want clusterings at all scales whose
distance distortion is $O(1)$, it is necessary to restrict our attention to classes of graphs that are somehow ``nice.''

\section{Low-Distortion Clusterings}
\label{sec:low-dist-clust}

In this section, we prove \cref{thm:distort-and-crossing}. Our main technical result is that when a maximal independent set of $G^{\le R}$ is used for the centers in a Voronoi clustering, the distance distortion is $O(1)$. Essentially, this follows from the following observation.

\begin{observation}\label{obs:cells}
    Let $\mathcal{C} = \vor(S)$, where $S$ is a maximal independent set of $G^{\le R}$. For any vertex $s \in S$, the Voronoi cell $\mathcal{C}(s)$ centered at $s$ satisfies
    \[
    \mathcal{B}_{R/2}(s) \subseteq \mathcal{C}(s) \subseteq \mathcal{B}_R(s).
    \]
\end{observation}
\begin{proof}
The containment $\mathcal{B}_{R/2}(s) \subseteq \mathcal{C}(s)$ follows from the fact that $\dist_G(u,v) > R$ for any two distinct vertices $u \in S$ and $v \in S$. The containment $\mathcal{C}(s) \subseteq \mathcal{B}_R(s)$ follows from the fact that for each vertex $u$ in $G$, there exists a vertex $v \in S$ such that $\dist_G(u,v) \leq R$.
\end{proof}

Informally, \cref{obs:cells} says that the inner and outer diameters of the Voronoi cells are uniformly within a constant factor of each other.

\begin{lemma} \label{lem:distortion-plain}
    Suppose $G^{\le R}$ is a bounded independence graph with parameters $\gamma, k$.
    Let $S$ be a maximal independent set for $G^{\le R}$.
    Let $\mathcal{C} = \vor(S)$ be the corresponding unweighted Voronoi clustering on $G$. For every $v,w \in V$, we have
    \[
    \left\lceil \frac{d(v,w)+1}{2R+1} \right\rceil \le d'([v],[w])+1 \le 
    \left\lceil \frac{d(v,w)+1}{2R+1} \right\rceil \gamma 2^k,
    \]
    where $d$ denotes shortest path distance in $G$, and $d'$ denotes shortest path distance in  $G / \mathcal{C}$.
\end{lemma}

\begin{proof}
First, \cref{obs:cells} implies that the diameter of every cluster is at most $2R$.
    Therefore, any path $P'$ of length $\ell'$ from $[v]$ to $[w]$ in  $G/\mathcal{C}$ lifts to a path $P$ from $v$ to $w$ in $G$, whose length is at most 
    \[
    2R(\ell' +1) +\ell'  = (2R +1)(\ell' +1) -1.
    \]
    To derive the length bound, we choose $P$ as a path that overlaps with the $\ell'+1$ clusters corresponding to the  $\ell'+1$ vertices in $P'$. The term $2R(\ell' +1)$ captures the fact that the length of the subpath of $P$ in each of the $\ell'+1$ clusters is at most $2R$. The term $\ell'$ captures the $\ell'$ edges in $P$ connecting the $\ell'+1$ clusters.

    Applying this to a shortest such path in $G/\mathcal{C}$, we have 
    \[
    d(v,w) + 1 \le (2R +1)(d'([v],[w])+1)
    \]
    Rearranging terms and recalling that $d'([v],[w])+1$ is an integer, we have 
    \[
     d'([v],[w])+1 \ge \left\lceil \frac{d(v,w) + 1}{2R+1}\right\rceil 
    \]

    For the other direction, \cref{obs:cells} implies that if $x,y$ are two vertices whose distance in $G$ is at most $R$, then $d(x,y_0) \le 2R$, where
    $y_0$ is the center of $[y]$. Suppose $v, w \in G$ are at distance $\ell$ and let $P$ be a shortest path between them. Partition $P$ into segments of $2R+1$ vertices. If $x$ is the central vertex of such a segment, then the centers of the clusters that intersect the segment all lie in a ball of radius $2R$ around $x$ in $G$.   Since $G^{\le R}$
    has bounded independence with parameters $\gamma, \ddim$, it follows that the number of such clusters
    is at most $\gamma 2^{\ddim}$.
    Since a shortest path of length $\ell$ can be partitioned into $\left\lceil (\ell+1)/(2R+1) \right \rceil$ segments
    of at most $2R+1$ vertices, each of which intersects at most $\gamma 2^{\ddim}$ clusters, we have constructed a path of length  $\left\lceil (\ell+1)/(2R+1) \right \rceil\gamma 2^{\ddim} -1$ from $[v]$ to $[w]$ in $G/\mathcal{C}$. It follows that 
    \[
    d'([v],[w])+1 \le \left\lceil \frac{d(v,w)+1}{2R+1} \right \rceil  \gamma 2^{\ddim},
    \]
    which completes the proof.
\end{proof}


The above lemma implies the desired constant distortion bound.

\begin{theorem} \label{thm:distortion-plain}
    Suppose $G^{\le R}$ is a bounded independence graph with parameters $\gamma, k$.
    For any maximal independent set $S$ of $G^{\le R}$, its corresponding Voronoi clustering $\mathcal{C} = \vor(S)$ of $G$ is a clustering at scale $R$ of constant distance distortion.
\end{theorem}
\begin{proof}
It follows from \cref{lem:distortion-plain} by rearranging terms, as $\gamma$ and $\ddim$ are both constants.
\end{proof}

Recall that our other goal in designing low-distortion clusterings is to ensure that the number of crossing edges is small. We cannot say this directly about the Voronoi clustering. However, we will show that this can be achieved by tweaking the clustering to be the \emph{additively weighted Voronoi clustering} with some small random weights. First, we show that the distortion remains constant even if we use small additive weights.

\begin{lemma}\label{lem:BI-distortion}
Suppose $G^{\le R}$ is a bounded independence graph with parameters $\gamma, k$. Let $S$ be a maximal independent set for $G^{\le R}$,
and let $W : S \to [0, C R]$ where $0 \le C < 1$.  
The clustering $\vor(S,W)$ has constant distance distortion at scale $R$.
\end{lemma}

\begin{proof}
    Since every vertex is within distance $R$ of a cluster center and the maximum value of
    $W$ is at most $CR$, all cluster diameters are at most $(1+C)R$.  Similar to the first half of the proof of \cref{lem:distortion-plain}, we have
    \begin{align*}
        d(v,w) &\le (d'([v],[w])+1)(1+C)R + d'([v],[w])\\
        & < (d'([v],[w])+1)((1+C)R +1), 
    \end{align*}
    which implies $(1+ d(v,w)/R )/(1+d'([v],[w])) = O(1)$.  
    
    Similar to the second half of the proof of \cref{lem:distortion-plain},
    each segment of $2R+1$ vertices in a shortest path intersects clusters whose centers are
    an independent set in $G^{\le R}$, lying within a ball of radius $(2+C)R$ 
    centered at the midpoint of the path. By the bounded independence of $G^{\le R}$, there are at most
    $\gamma (2+C)^k$ of them.  This implies $(1+d(v,w)/R)/(1+d'([v],[w])) = \Omega(1)$,
    completing the proof.
\end{proof}

If we combine the above observation about distance distortion with the idea
of using \emph{random} additive weights to perturb the boundaries of the clusters, 
we can simultaneously guarantee low distortion and few crossing edges, at least in expectation.

\begin{theorem} \label{cor:distort-and-crossing}
    Suppose $G^{\le R}$ is a bounded independence graph with parameters $\gamma, k$. Let $S$ be a maximal independent set for $G^{\le R}$
and let $W : S \to [0, R/10]$ be chosen uniformly at random.  The clustering $\mathcal{C} = \vor(S,W)$ has constant distance distortion  at scale $R$.
Moreover, any edge $e \in E$ has probability $O(1/R)$ to be a crossing edge under clustering $\mathcal{C}$.
\end{theorem}

\begin{proof}
The constant distance distortion bound follows from \cref{lem:BI-distortion}.
    To see that the probability of an edge, $\{u,v\}$, crossing a cluster boundary is $O(1/R)$, we first observe that, by the bounded independence
    property, there are at most $O(1)$ cluster centers within distance $1.1R$ of $u$ or $v$.  Consider one such pair of cluster centers, $w,x$ 
    and consider the event that $u$ joins $w$'s cluster and $v$ joins $x$'s cluster.   
    For this to happen, we must have $d(u,w)-W(w) < d(u,x)-W(x)$ and $d(v,x)-W(x) <  d(v,w)-W(w)$.
    Since $u$ and $v$ are neighbors, combining these constraints with the triangle inequality implies
    $|W(w) - W(x) + d(u,x) - d(u,w)| \le  2$.  Since the weight $W(w)$ is uniform over a range of length $R/10$,
    the probability of this event is $O(1/R)$, even conditioned on the value of $d(u,x)-d(u,w)-W(x).$
    Taking a union bound over the $O(1)$ possible choices of pairs of centers, we conclude that the probability of $\{u,v\}$ becoming a crossing edge is $O(1/R)$.
\end{proof}

 \cref{cor:distort-and-crossing} implies our main \emph{existential} result of the paper.

\thmDistCross*
\begin{proof}
   The constant-distortion clustering  $\mathcal{C}$ of \cref{cor:distort-and-crossing} has  $O(1/R)$ fraction of crossing edges \emph{in expectation}. By Markov inequality, with probability $1/2$, the fraction of crossing edges is $O(1/R)$, satisfying the requirements in the theorem statement.
\end{proof}

\subsection{Distributed Clustering Algorithms }\label{sec:build-clust}

For our applications, we need distributed algorithms that enable the vertices in a communication network to self-organize into clusters achieving the nice properties  discussed above. In particular, we need a method to select the cluster centers.  Fortunately, the required
algorithmic tools have already been developed.  Schneider and Wattenhofer~\cite{schneider2010optimal} showed an $O(\logstar n)$-round $\CONGEST$
algorithm, hereafter called the SW algorithm, for finding a maximal independent set in a graph of bounded independence.  There is a slight difficulty in deploying  their algorithm 
to find the centers for our  Voronoi  clustering, in that we want a maximal independent set not on  the original graph $G$,
but on the power graph $G^{\le R}$, where $R$ is the scale we are going for. If $G$ has uniformly bounded independence, then $G^{\le R}$ is amenable to running the SW algorithm. However, the communication for this algorithm needs to be simulated on the underlying graph $G$. 

\subparagraph{Finding Cluster Centers in the $\LOCAL$ Model.} 

In the $\LOCAL$ model, vertices are able to share their complete state, so that in $t$ rounds, each vertex can know everything about its $t$-hop neighborhood. Since the SW algorithm runs for $O(\logstar n)$ steps on a graph of bounded independence, its actions can only depend on vertices that are within that distance. Thus working in the $\LOCAL$ model, if the vertices of $G$ share state out to distance $O(R\logstar n)$, then each vertex can simulate the SW algorithm on $G^{\le R}$ on its own, and determine whether or not it is in the MIS, \emph{i.e.}, whether it is a cluster center.

\subparagraph{Main Clustering Algorithms.}
In \Cref{alg:voronoi,alg:voronoi-rand}, we describe the distributed algorithms used to build the Voronoi clusters and the weighted Voronoi clusters at scale $R$, respectively. Both algorithms start by using the Schneider--Wattenhofer algorithm to find a maximal independent set in $G^{\le R}$ to use as cluster centers. Each cluster center initiates an $R$-depth BFS rooted at itself. Each vertex joins a cluster when the BFS tree from some cluster center reaches it. The maximality of the set of cluster centers ensures that this process does reach every vertex. The only difference between the two algorithms is that in the first version, all the BFS runs start at the same time ($t=0$), whereas in the weighted version, each cluster center chooses a small random delay (between $0$ and $R/10$) to start the BFS.

Since the BFS part only goes out to distance $O(R)$, the dominant term in the round complexity comes from the simulation of the Schneider--Wattenhofer algorithm, and we have proved the following result.

\begin{theorem}\label{thm:implvor}
Let $G$ be a graph such that $G^{\le R}$ has bounded independence. 
\begin{itemize}
    \item \Cref{alg:voronoi} is a deterministic algorithm in the $\LOCAL$ model that costs  $O(R\logstar n)$ rounds and produces an $O(1)$-distortion clustering at scale $R$ with clusters of diameter at most $2R$.
    \item \Cref{alg:voronoi-rand}  is a randomized algorithm in the $\LOCAL$ model that costs   $O(R\logstar n)$ rounds and produces an $O(1)$-distortion clustering at scale $R$ with clusters of diameter at most $2.2R$. Moreover, each edge is a crossing edge with probability $O(1/R)$, so the clustering has an $O(1/R)$ fraction of crossing edges in expectation.
\end{itemize}
\end{theorem}
\begin{proof}
   The round complexity $O(R\logstar n)$ for both algorithms is due to the implementation above. The constant distortion bound for both algorithms is due to \Cref{lem:BI-distortion}. For \Cref{alg:voronoi-rand}, the probability $O(1/R)$ for an edge to be a crossing edge is due to \Cref{cor:distort-and-crossing}.
   For each cluster center $s$ in \Cref{alg:voronoi}, the cluster $\mathcal{C}(s)$ centered at $s$ satisfies $\mathcal{C}(s) \subseteq \mathcal{B}_R(s)$, so the diameter of $\mathcal{C}(s)$ is at most $2R$. For \Cref{alg:voronoi-rand}, the corresponding guarantee is $\mathcal{C}(s) \subseteq \mathcal{B}_{1.1R}(s)$ due to the additive weights that can be as large as $R/10$, so the cluster diameter bound becomes $2.2R$.
\end{proof}

\Cref{alg:voronoi,alg:voronoi-rand} are \emph{energy-inefficient}: In the implementation described above, both the round complexity and energy complexity are $O(R\logstar n)$. Later, in \Cref{sec:energy}, we will show that, in fact, \Cref{alg:voronoi,alg:voronoi-rand} can be implemented using only $O(\log R \logstar n)$ energy, which is an \emph{exponential} improvement.

The following result is a special case of \cref{thm:implvor}.

\thmLOCALclustering*

\begin{algorithm} \caption{Voronoi clustering with scale $R$.} \label{alg:voronoi}
\begin{algorithmic}
\State{First run the SW algorithm in $G^{\le R}$ to select the cluster centers. }
\State{Perform BFS out to distance $R$ from each cluster center.}
\State{Each vertex that joins the BFS tree records its cluster center's ID, its level in the BFS tree, and its parent in the BFS tree. Then it sends its state to all its neighbors other than its parent. Each such neighbor that has not previously joined a cluster can now join the cluster at the next level of the BFS tree. }
\end{algorithmic}
\end{algorithm}

\begin{algorithm} \caption{Weighted Voronoi clustering with scale $R$ and random start times.} \label{alg:voronoi-rand}
\begin{algorithmic}
\State{First run the SW algorithm in $G^{\le R}$ to select the cluster centers. }
\State{Each cluster center picks a start time uniformly at random from $[0, R/10]$}
\State{Perform BFS out to distance $R$ from each cluster center, starting at the sampled start times.}
\State{Each vertex that joins the BFS tree records its cluster center's ID, its level in the BFS tree, and its parent in the BFS tree. Then it sends its state to all its neighbors other than its parent. Each such neighbor that has not preciously joined a cluster can now join the cluster at the next level of the BFS tree.}
\end{algorithmic}
\end{algorithm}

\subsection{Derandomizing the Clustering}

In \Cref{thm:implvor}, while both \Cref{alg:voronoi,alg:voronoi-rand} have low distance distortion, only the randomized one guarantees few crossing edges \emph{in expectation}. In this section, we show that a deterministic algorithm attaining both guarantees can be obtained by applying the method of conditional expectations to derandomize the selection of additive weights in \Cref{alg:voronoi-rand}. 

\begin{theorem}\label{thm:derand}
There is a deterministic $\LOCAL$ algorithm which, when $G^{\le R}$ is bounded independence,
produces an $O(1)$-distortion clustering at scale $R$ satisfying the following requirements.
\begin{itemize}
    \item For each cluster center $s$, its cluster $\randclust(s)$ satisfies $\nbd{0.4R}{s} \subseteq \randclust(s) \subseteq \nbd{1.1R}{s}$.
    \item The clustering has an $O(1/R)$ fraction of crossing edges.
\end{itemize}
The round complexity of this algorithm is $O(R \logstar n)$.
\end{theorem}
\begin{proof}
First of all, observe that regardless of the choice of the start times in $[0, R/10]$ in \Cref{alg:voronoi-rand}, for each cluster center $s$, its cluster $\randclust(s)$ satisfies $\nbd{0.4R}{s} \subseteq \randclust(s) \subseteq \nbd{1.1R}{s}$.

Consider the graph $G^*=(V^*, E^*)$ where $V^*$ is the independent set that has been selected to be the centers of the clusters, and  $\{u,v\} \in E^*$ if $d_G(u,v) \le 2.2R + 1$. The choice of the threshold $2.2R + 1$ guarantees that $\{u,v\} \in E^*$ when the clusters of the $u$ and $v$ \emph{could be adjacent} for some choice of start times, as $\randclust(s) \subseteq \nbd{1.1R}{s}$ for each cluster center $s$. From now on, we say that $u \in V^*$ and $v \in V^*$ are \emph{competing} if $\{u,v\} \in E^*$. Our definition of $G^*$ ensures that for any collection $S \subseteq V^\ast$ of mutually non-competing centers, the event for any edge $e \in E$ to be a crossing edge can be affected by \emph{at most one} center in $S$. 

We plan to set the start time for each cluster center in $V^*$ in \Cref{alg:voronoi-rand} \emph{sequentially}. When we process one cluster center $s \in V^*$, we set its start time deterministically to \emph{minimize} the conditional expected value of the number of crossing edges. To do so, $s$ just needs to learn the graph topology of the ball of radius $2.2R + 1$ around $s$ and the start times of all \emph{competing} centers who have fixed their start times. To summarize, for $s$ to set its start time, $s$ just needs to gather information within a ball of radius $O(R)$, which can be done in $O(R)$ rounds in the $\LOCAL$ model.

While the above plan seems inherently sequential, a key observation is that any collection $S \subseteq V^\ast$ of mutually non-competing centers can set their start times in parallel, allowing us to parallelize the procedure.
Since $G$ has uniformly bounded independence, $G^*$ has a constant maximum degree. 
Thus we can use Linial's algorithm~\cite{Linial92} to properly color $G^*$  using $O(1)$ colors in $O(\logstar n)$ rounds in $G^*$. Moreover, this can be simulated in the underlying graph $G$ in 
$O(R\logstar n)$ rounds in $\LOCAL$. 
Now we go through each color class one by one: Since centers in the same color class are mutually non-competing, they can choose their start times in parallel in $O(R)$ rounds. 

In the end, we obtain a clustering whose fraction of crossing edges is \emph{at most} the \emph{expected} fraction of crossing edges of the clustering of \Cref{alg:voronoi-rand}, which is $O(1/R)$ by \Cref{thm:implvor}.
    Finally, the constant distortion bound is due to \Cref{lem:BI-distortion}.   
\end{proof}

The following theorem is a simplification of \Cref{thm:derand}.

\thmLOCALclusteringDet*
\begin{proof}
The theorem follows from \Cref{thm:derand} immediately.
\end{proof}

\section{Low-Energy Simulation in the \texorpdfstring{$\LOCAL$}{LOCAL} Model}

\label{sec:energy}
In this section, we show that for uniformly bounded independence graphs, any $t \logstar n$-round $\LOCAL$ algorithm can be simulated using only $O(\log t \logstar n)$ energy per vertex, without affecting the asymptotic round complexity of the algorithm! This is achieved by following the very high-level idea of Dani and Hayes~\cite{dani2022wake}  that \emph{multi-scale} low-distortion clusterings can be used to do \emph{generic} simulation of existing algorithms using as little energy as possible.  

Once we have a constant-distortion clustering at scale $R$, using the BFS trees associated with the clusters, in the 
$\LOCAL$ model, we can let each node gather all information within its ball of radius $CR$ using $O(CR)$ rounds and $O(C)$ energy. Such clustering enables a generic \emph{low-energy simulation} in the $\LOCAL$ model: Any $CR$-round algorithm can be simulated using $O(CR)$ rounds and $O(C)$ energy! In particular, this allows us to construct a constant-distortion clustering at scale $2R$ using $O(R \logstar n)$ rounds and only $O(\logstar n)$ energy. Therefore, by starting from the trivial clustering of scale $1$ and then building constant-distortion clusterings at scales $2, 4, 8, 16, \ldots, 2^{\lceil \log R \rceil}$, we can drastically improve the energy complexity of \Cref{thm:LOCAL-clustering-det} from $O(R \logstar n)$ to just $O(\log R \logstar n)$!

\subsection{Bootstrapping Multi-Scale Voronoi Clusters}

In this section, we show that for uniformly bounded independence graphs, a constant-distortion clustering at scale $2^i$ can be computed using $O(2^i \logstar n)$ rounds and $O(i \logstar n)$ energy. The algorithm involves building constant-distortion clusterings at scale $2^0, 2^1, 2^2, \ldots, 2^i$. In the subsequent discussion, we write $\mathcal{C}_j$ to denote the constant-distortion clustering at scale $2^j$. Each cluster of $\mathcal{C}_j$ is called a cluster of \emph{level} $j$.

Suppose, for an induction hypothesis, that we have already built level $j$ clusters, where $j \ge 0$. That is, $\mathcal{C}_j$ is a constant-distortion clustering guaranteed by \Cref{thm:LOCAL-clustering-det}. Recall that by definition of a clustering (\cref{def:clustering}), each vertex knows the center of its cluster, its distance from that center, and its parent in the BFS tree rooted at the center. Furthermore, since we are in the $\LOCAL$ model, by doing one round of communication, we can assume that each vertex $v$ knows its children in a BFS tree rooted at the cluster center of $[v]$ and knows  which neighbors of $v$ are in the same cluster $[v]$.
We want to use the level $j$ clusters to build the next level of clusters with low energy usage. In order to do that we will need certain energy-saving primitives in $\mathcal{C}_j$.
\begin{description}
    \item[$\DOWNCAST$:] The cluster center broadcasts its information to the entire cluster.
    \item[$\UPCAST$:] The cluster center gathers all information within the cluster.
    \item[$\INTERCAST$:] Neighboring clusters exchange their information.
\end{description}

These procedures are described more formally in \Cref{alg:cluster-local}. 

\begin{algorithm} \caption{Basic cluster operations.} \label{alg:cluster-local}
\begin{algorithmic}
\State{ }\Comment{For each of the procedure below, we require that either all the vertices in a particular cluster participate, or none do.} 
\State{Let $t_0$ be the round number when the procedure starts.}
\State{Let $\delta_v$ be the distance from $v$ to the center of $[v]$.}
\State{Let $d$ be a known bound on the maximum cluster diameter.}
\Procedure{Downcast}{$t_0$} 
\State{The cluster center wakes up at time $t_0$ and sends its current state to all its neighbors.} 
\State{Every non-center vertex $v$ wakes up at time $t_0 + \delta_v -1$ to receive its parent's state, which $v$ appends to its own state. At time $t_0 + \delta_v$, $v$ sends its current state to its children and then goes back to sleep.}
\EndProcedure
\Procedure{Upcast}{$t_0$} 
\State{Each leaf vertex $v$ wakes up at time $t_0 + d -\delta_v$ and sends its current state to its parent, and then it goes back to sleep.} 
\State{Each non-leaf vertex $v$ wakes up at time $t_0 +d - \delta_v - 1$ and receives the state of its children, which it appends to its own state. If $v$ is not the center, at time $t_0 +d - \delta_v$, $v$ sends its current state to its parent and then goes
back to sleep.}
\EndProcedure
\Procedure{Intercast}{$t_0$} 
\State{Each vertex $v$ wakes up at time $t_0$, sends its current state to all its neighbors in other clusters, and receives their current state, which $v$ appends to its own.} 
\EndProcedure
\end{algorithmic}
\end{algorithm}

\begin{lemma}
    For level $j$ clusters, $\DOWNCAST$ and $\UPCAST$ cost $O(2^j)$ rounds, and $\INTERCAST$ cost one round. Moreover, all these three operations use at most two energy units per vertex. 
\end{lemma}
\begin{proof}
    $\DOWNCAST$ and $\UPCAST$ take exactly $d = O(2^j)$ rounds because they correspond to information traveling once down the BFS tree and once up it, respectively. Since each vertex in the cluster knows its distance from the center, it only wakes up for the two rounds when it receives and sends messages. For $\DOWNCAST$, that is the $(i-1)$th and the $i$th rounds after the procedure starts, where $i$ is its distance from the center. For $\UPCAST$, the two awake rounds are set to $d - i$ and $d -i +1$ from the start time, where $i$ is its distance from the center. Thus $\DOWNCAST$ and $\UPCAST$ use only two units of energy per vertex. $\INTERCAST$ only runs for one timestep, so clearly it only uses one unit of energy per vertex. 
\end{proof}

\begin{algorithm} \caption{Construction of level $j+1$ clusters.} \label{alg:bootcluster-local}
\begin{algorithmic}
\State{ }
\Comment{Assumes that Level $j$ clusters already exist.}
\State{Let $t_0$ be the round number when the procedure starts.}
\State{Let $M= O(\logstar n)$ be a known value.}
\State{Let $d = O(2^j)$ be a known bound on the maximum cluster diameter in $\mathcal{C}_j$.}
\Procedure{BuildCluster}{$t_0$} 
\State{Run $\DOWNCAST$($t_0$)}
\State{Run $\UPCAST$($t_0 + d)$}
\State{$\tau \gets t_0 + 2 d$}
\For{$i \gets 1 \mbox{ to } M$}
\State{Run $\DOWNCAST$($\tau$)}
\State{Run $\INTERCAST$($\tau + d$)}
\State{Run $\UPCAST$($\tau + d + 1)$}
\State{$\tau \gets t_0 + 2 d +1 $}
\EndFor
\State{Run $\DOWNCAST$($\tau$)}
\State{Locally simulate a Voronoi-cluster-forming algorithm to determine your cluster center and BFS tree for level $j+1$ clustering.}
\EndProcedure
\end{algorithmic}
\end{algorithm}






Given $\mathcal{C}_j$, to build level $j+1$ clusters, we want to run the Voronoi clustering algorithm (\Cref{thm:LOCAL-clustering-det}) with scale $2^{j+1}$, by first sharing state out to distance $O(2^{j+1} \logstar n)$ in $G$ by repeated iterations of $\DOWNCAST$-$\INTERCAST$-$\UPCAST$ and then each vertex simulating the algorithm locally. The details are in \Cref{alg:bootcluster-local}, where we select $M = O(\logstar n)$ to be sufficiently large so that for each vertex $v \in V$, the ball of radius $M$ around $[v]$ in the cluster graph of $\mathcal{C}_j$ \emph{covers} the ball of radius $t = O(2^{j+1}\logstar n)$ around $v$ in $G$, where $t$ is the round complexity of the clustering algorithm of \Cref{thm:LOCAL-clustering-det} with scale $2^{j+1}$. This is possible because $\mathcal{C}_j$ is a constant-distortion clustering with scale $2^j$ by the induction hypothesis.

\begin{lemma}\label{lem:level-j}
Given level $j$ clustering $\mathcal{C}_j$, level $j+1$ clustering $\mathcal{C}_{j+1}$ can be constructed using  $O(2^{j+1} \logstar n)$ rounds and $O(\logstar n)$ energy.
\end{lemma}
\begin{proof}
    After one initial $\DOWNCAST$-$\UPCAST$ pair, all the centers know the state of their entire cluster. Thereafter, we execute $M = O(\logstar n)$ iterations of $\DOWNCAST$-$\INTERCAST$-$\UPCAST$. Each such iteration results in the current state being exchanged between a cluster center and all its neighboring cluster centers. Thus, after $M$ such iterations, the cluster centers have exchanged state with all cluster centers that are within $M$ hops of themselves in the cluster graph, and an additional $\DOWNCAST$ results in everyone in a cluster having the same information as the center. 
    By our choice of $M$, for any vertex $v \in V$, the ball of radius $t = O(2^{j+1}\logstar n)$ centered at 
    $v$ is entirely within distance $M = O(\logstar n)$ of  $[v]$ in the level $j$ cluster graph, where $t$ is the round complexity of the clustering algorithm of \Cref{thm:LOCAL-clustering-det} with scale $2^{j+1}$.  
    Thus after $O(2 \logstar n)$ iterations of $\DOWNCAST$-$\INTERCAST$-$\UPCAST$ followed by one $\DOWNCAST$, all the vertices know the state up to distance $t$, and can therefore locally simulate the clustering algorithm of \Cref{thm:LOCAL-clustering-det} with scale $2^{j+1}$ to compute the level $j+1$ clustering $\mathcal{C}_{j+1}$. Since the cluster operations only use constant energy, it follows that if we have already built the clusters at scale $2^{j}$, then we can build the clusters at scale $2^{j+1}$  using  $O(2^{j+1} \logstar n)$ rounds and $O(\logstar n)$ energy.
\end{proof}

\begin{theorem}\label{thm:bootstrapping}
Let $i > 0$ be an integer. A constant-distortion clustering at scale $2^i$ for a uniformly bounded independence graph can be computed using $O(2^i \logstar n)$ rounds and $O(i \logstar n)$ energy in the $\LOCAL$ model deterministically.
\end{theorem}
\begin{proof}
For the base case of $i = 0$, we take the trivial clustering: $[v] = \{v\}$ for all $v \in V$.
For the inductive step, by \Cref{lem:level-j}, building level $j+1$ clusters from level $j$ clusters cost $O(2^{j+1} \logstar n)$ rounds and $O(\logstar n)$ energy.
Summing this up over the $i$ levels, we have a round complexity of 
\[
O(\logstar n) \cdot \left(2^1 + 2^2 + 2^3 + \cdots + 2^{i}\right) = O(2^i \logstar n)
\]
and an energy usage of $O(i \logstar n)$.     
\end{proof}

\subsection{Overlapping Clusters}
In the proof of \Cref{lem:level-j}, we actually show that \Cref{alg:bootcluster-local} allows us to simulate an \emph{arbitrary} $O(RM)$-round $\LOCAL$ algorithm using $O(RM)$ rounds and $O(M)$ energy, given that a constant-distortion clustering at scale $R$ has been constructed. This observation, together with \Cref{thm:bootstrapping}, already allows us to prove \Cref{thm:localsim}! However, we can do it slightly better by constructing a combinatorial structure that can be more conveniently applied to the generic low-energy simulation.

The Voronoi clustering \emph{partitions} the graph into disjoint clusters, each of which contains a ball of radius $R/2$ and is contained in a ball of radius of $R$. There are contexts---for example, simulation of $t$-step $\LOCAL$ algorithm---in which it is convenient to be able to say that \emph{every} ball of some size is contained in a cluster. This is, of course, impossible if the clusters form a partition, but if one is allowed to relax that so that instead they form a covering of the graph, then it is trivially possible, but not interesting, by taking the covering by all balls. The happy medium is when we can achieve this goal while keeping the overlap in the covering small in some sense. The relevant property here is the notion of an \emph{$\ell$-fold cover}: Any vertex is covered at most $\ell$ times. 

\begin{definition}\label{def:overlap}
Given a graph $G=(V,E)$, a $d$-diameter $\ell$-fold {cover} of balls of radius $t$ is a collection $\mathcal{C}$ of clusters satisfying the following conditions.
\begin{itemize}
    \item For each cluster $C \in \mathcal{C}$, the diameter of the subgraph induced by $C$ is at most $d$
    \item For each vertex $v \in V$, $\mathcal{B}_t(v) \subseteq C$ for some cluster $C \in \mathcal{C}$.
    \item For each vertex $v \in V$, the number of clusters in $\mathcal{C}$ that contains $v$ is at most $\ell$.
\end{itemize}
\end{definition}

Here we also assume that each cluster $C \in \mathcal{C}$ in the above definition is associated with a center and a BFS tree rooted at the center in the same way as \Cref{def:clustering}.


We prove the existence of such a covering for graphs with uniformly bounded independence.

\begin{lemma}\label{lem:overlap}
    Let $G$ be a graph with uniformly bounded independence, $R \ge 1$, and  $S = \{v_1, \dots, v_k\}$ a maximal independent set in $G^{\le R}$.
    The collection of balls 
    \[
\mathcal{C} = \{ \nbd{3R}{v_i} \, : \, 1\le i\le k \}
\]
satisfies the following conditions for every $w \in V$.
\begin{itemize}
    \item [(a)] There exists $1 \le i \le k$ such that $\nbd{R}{w} \subseteq \nbd{2R}{v_i}$, and 
    \item [(b)] There are only $O(1)$ values of $1 \le j \le k$ such that $w \in \nbd{2R}{v_j}$.
\end{itemize}
Therefore, the collection $\mathcal{C}$ of balls forms a $6R$-diameter $O(1)$-fold cover of balls of radius $R$.
\end{lemma}

\begin{proof}
    Let $w$ be any vertex and $v_i$ be the center of its cluster in $\vor(S)$, then $w \in \nbd{R}{v_i}$, implying that $\nbd{R}{w} \subseteq \nbd{2R}{v_i}$, establishing (a). To see (b), observe that  if $w \in \nbd{2R}{v_j}$, then $v_j \in \nbd{2R}{w}$, but by the uniformly bounded independence of $G$, at most $O(1)$ of $v_j \in S$ fall so close to $w$.
\end{proof}


\begin{theorem}\label{thm:overlap}
For any given parameter $R \geq 1$, an $O(R)$-diameter $O(1)$-fold cover of balls of radius $R$ for a uniformly bounded independence graph can be constructed using $O(R \logstar n)$ rounds and $O(\log R \logstar n)$ energy.
\end{theorem}
\begin{proof}
By \Cref{lem:overlap}, once we compute a maximal independent set $S$ of $G^{\le R}$, a desired $O(1)$-fold cover can be computed in $O(R)$ rounds. By simulating the SW algorithm in $G^{\le R}$, a maximal independent set $S$ of $G^{\le R}$ can be computed in $O(R \logstar n)$ rounds. 

To improve the energy complexity to $O(\log R \logstar n)$, we first run the algorithm of \Cref{thm:bootstrapping} with $i = \lceil \log R \rceil$ to obtain a constant-distortion clustering with scale $2^i$. After that, following the proof of \Cref{lem:level-j}, we can use \Cref{alg:bootcluster-local} to simulate the above $O(R \logstar n)$-round $\LOCAL$ algorithm using $O(R \logstar n)$ rounds and $O(\logstar)$ energy, as $R = O(2^i)$. The dominant term for the energy complexity is the algorithm of \Cref{thm:bootstrapping}, which costs $O(\log R \logstar n)$ energy.
\end{proof}

\subsection{Simulating the Algorithm}




Now we are ready to prove \Cref{thm:localsim}.

\thmLOCALsim*

\begin{proof}
Using the algorithm of \Cref{thm:overlap}, we compute an $O(t)$-diameter $O(1)$-fold cover of balls of radius $t$ using $O(t \logstar n)$ rounds and $O(\log t \logstar n)$ energy.
In the $O(1)$-fold cover, the ball of radius $t$ around each vertex is contained in some cluster. Therefore, one iteration of $\DOWNCAST$-$\UPCAST$-$\DOWNCAST$ results in each vertex knowing the state of its entire $t$ neighborhood, so that it can simulate the $t$-round $\LOCAL$ algorithm $\mathcal{A}$ on its own.
In the $O(1)$-fold cover, each vertex belongs to $O(1)$ clusters, so running $\DOWNCAST$-$\UPCAST$-$\DOWNCAST$ for all clusters in parallel still cost $O(t)$ rounds and $O(1)$ energy. To simulate the given $t \logstar n$-round algorithm $\mathcal{A}$, we just need to repeat $\DOWNCAST$-$\UPCAST$-$\DOWNCAST$ for $\logstar n$ iterations, costing $O(t \logstar n)$ rounds and $O(\log t \logstar n)$ energy, as required.
\end{proof}

An interesting consequence of \Cref{thm:localsim} is that the energy complexity of \Cref{thm:LOCAL-clustering} can be automatically improved to $O(\log R \logstar n)$ without affecting the asymptotic round complexity $O(R \logstar n)$. 

\section{Low-Energy Simulation in \texorpdfstring{$\CONGEST$}{CONGEST}}
\label{sec:newCONGESTstuff}

Dani and Hayes~\cite{dani2022wake} showed that, up to polylogarithmic factors in the 
round and energy complexities, 
there is a ``best'' generic method for reducing the
energy cost of any algorithm in the $\RADIO$ model.
More specifically, they show that their simulation algorithm
always performs at least as well as every algorithm in a 
class of black-box simulation algorithms, that they call 
\emph{safe one-pass generic simulation algorithms.}
Their results also hold with trivial modifications in the $\CONGEST$ model.

Their method is based on using the clustering algorithm of Miller, Peng and Xu~\cite{MPX} to
build approximately distance-preserving clusters at all scales, and then
use these clusters as the basis for a kind of ``warning'' system that allows vertices
to sleep most of the time, and wake up only when there is the possibility that nearby
vertices will send messages soon.  

Since our new clustering method has less distance distortion than MPX clustering, it is
natural to ask whether it can be used to improve the results of~\cite{dani2022wake}, and
yes it can!  Moreover, since their algorithm (called ``$\SAF$ simulation'') treats its
clusterings as black boxes, plugging our clusterings in yields an immediate improvement,
reducing both the round and energy complexities by about a $\log^2(n)$ factor.  

In this section, we present a few more details of how this works.  As observed above,
once the clusterings have been built, the algorithms work as designed.  This applies
in both the $\CONGEST$ and the $\RADIO$ models.  However, there is still the matter
of building the clusterings in the first place, using as little energy as possible.
This is the main technical challenge we will address in this section.


We will now state the main results.  As mentioned above, they are closely based on the main theorem
of~\cite{dani2022wake}, but stated for $\CONGEST$ as well as $\RADIO$.  In $\CONGEST$, the
round and energy complexities are reduced by a $\log(n) \log(\Delta)$ factor compared with $\RADIO$
because there is no issue of collisions causing messages to be dropped.  A further $\log^2(n)$
factor decrease in round and energy complexities comes from the reduced distance distortion of our clusterings. 
First, the version for $\CONGEST$:

\begin{theorem} \label{thm:main-congest}
Let $G = (V,E)$ be a graph on $n$ vertices of diameter $D$.
For every $\CONGEST$ algorithm, $\mathcal{A}$, there is a randomized $\CONGEST$ algorithm $\SAF(\mathcal{A})$, whose round complexity is
$O(\log(D) \TIME(\mathcal{A}))$, with the following properties.

\noindent
{\bf Precondition:} Assume all vertices are initialized with their local view of a hierarchical clustering for $G$ with distance distortion $O(1)$.

Under the above precondition, the simulation achieves the following two properties:
\begin{enumerate}
    \item With probability $1 - 1/\poly(n)$, $\SAF(\mathcal{A})$ produces the same output as $\mathcal{A}$. 
    \item Moreover, its energy cost satisfies, for every vertex $v$,
\[
\SAL(\SAF(\mathcal{A}),v) = O(\log(D) \OPT(\mathcal{A},v)),
\]
where $\OPT$ is the least possible energy cost for a safe, synchronized, one-pass generic $\CONGEST$ simulation algorithm.
\end{enumerate}
Furthermore, the above precondition can be guaranteed without assumptions, for a one-time cost of
$O(D \log^2 D)$ rounds and $O(\log^2 D)$ maximum energy cost per vertex.
\end{theorem}

And next, the version for $\RADIO$:

\begin{theorem}  \label{thm:main-radio}
Let $G = (V,E)$ be a graph on $n$ vertices of maximum degree $\Delta$ and diameter $D$.
For every $\RADIO$ algorithm, $\mathcal{A}$, there is a randomized $\RADIO$ algorithm $\SAF(\mathcal{A})$, whose round complexity is
$O(\log(D)\log(n)\log(\Delta) \TIME(\mathcal{A}))$, with the following properties.

\noindent
{\bf Precondition:} Assume all vertices are initialized with their local view of a hierarchical clustering for $G$ with distance distortion $O(1)$.

Under the above precondition, the simulation achieves the following two properties:
\begin{enumerate}
    \item With probability $1 - 1/\poly(n)$, $\SAF(\mathcal{A})$ produces the same output as $\mathcal{A}$. 
    \item Moreover, its energy cost satisfies, for every vertex $v$,
\[
\SAL(\SAF(\mathcal{A}),v) = O(\log(D)\log(n)\log(\Delta) \OPT(\mathcal{A},v)),
\]
where $\OPT$ is the least possible energy cost for a safe, synchronized, one-pass generic $\RADIO$ simulation algorithm.
\end{enumerate}
Furthermore, the above precondition can be guaranteed without assumptions, for a one-time cost of
$O(D \log^2 (D) \log (n) \log (\Delta))$ rounds and $O(\log^2 (D) \log (n) \log (\Delta))$ maximum energy cost per vertex.
\end{theorem}

It is worth mentioning that the above result applies to arbitrary algorithms $\mathcal{A}$.  In the special case when $\mathcal{A}$ is already using
Backoff at every step to deal with the possibility of collisions, there is an easy modification to the $\SAF$ simulation protocol that
correspondingly reduces the rate of running $\WARN$, and hence does not incur the extra $\log(n)\log(\Delta)$ factor of overhead to round and energy complexities.




\subsection{\texorpdfstring{$\CONGEST$}{CONGEST}}

We first prove a couple of lemmas about simulating $\CONGEST$ algorithms when the graph on which the algorithm is to be run (\emph{i.e.,} the input graph) is not the same as the underlying communication network.

\begin{lemma}\label{lem:CONG-sim-shortpath}
    Let $G=(V,E)$ be a graph with maximum degree $\Delta = O(1)$ and let $V' \subseteq V$. Suppose $G' = (V', E')$ is a graph such that 
    \begin{itemize}
        \item the maximum degree in $G'$ is $\Delta' = O(1)$ and
        \item every pair of adjacent vertices in $G'$ have distance at most $C= O(1)$ in $G$. 
    \end{itemize}  
    Moreover, assume that every vertex in $V$ knows whether or not it is an element of $V'$, and that every vertex in $V'$ knows the unique $\ID$ of each of its neighbors in $G'$.
    Then any $\CONGEST$ algorithm on $G'$ can be simulated by a $\CONGEST$ algorithm on $G$ with at most a 
    factor $\Delta' \Delta^C = O(1)$ increase in round complexity.
\end{lemma}

\begin{proof}
    We will show that $\Delta' \Delta^C$ rounds of communication in $G$ suffice to simulate a single round of communication of a $\CONGEST$ algorithm $\mathcal{A}$ on $G'$.  In one round of $\mathcal{A}$, each vertex $v \in V'$ may want to send up to $\Delta'$ messages, one to each of its neighbors. Although it knows its neighbors $\ID$s, and also that they are within $C$ hops away, it does not know the shortest paths to them. vertex $v$ tags each of its messages with the $\ID$ of the intended recipient and floods each of them out to depth $C$, $\Delta^C$ rounds apart. Thus we need to show that all the messages sent by vertices in $V'$ (one from each such vertex) arrive at their destinations within $\Delta^C$ rounds. Now consider any vertex $w \in V$. In any particular round, $w$ can receive at most $\Delta$ messages, one from each neighbor. Consider a round $t_0$ during which all vertices in $V'$ initiate a flooding operation. We will divide the time interval starting in round $t_0$ into phases of lengths 1, $\Delta$, $\Delta^2, \dots \Delta^{C-1}$. Any messages received by vertex $w \in V$  in phase $i$ will be forwarded in phase $i+1$. Then, the message originating at vertex in $v \in V'$ has reached all its neighbors in phase 0, and by induction, has reached everyone at distance $i$ from it by the end of phase $i-1$. Since the intended recipient is at a distance at most $C$ from it, the message reaches by the end of phase $C-1$. Finally, since 
    \[
    \sum_{i=0}^{C-1} \Delta^i  = \frac{\Delta^C -1}{\Delta -1} < \Delta^C,
    \]
    it follows that all messages arrive in the specified interval, and we've shown that $\Delta' \Delta^C$ rounds suffice to simulate one round of communication of $\mathcal{A}$. 
\end{proof}

\begin{lemma}\label{lem:CONG-sim}
    Let $G = (V,E)$ be a graph, and let $\mathcal{C}$ be a clustering of $V$ so that all vertices in $G$ know the $\ID$ of their parent vertex in their cluster, the $\ID$
    of their cluster center, and their distance to the cluster center. 
    Let $G' = G / \mathcal{C}$ be the cluster graph.
    Assume $G'$ has maximum degree $\Delta' =O(1)$ and constant distance distortion $C$ at some scale $R$.
    Then there exists a constant $C'$ such that every $\CONGEST$ algorithm on $G'$ can be simulated
    by a $\CONGEST$ algorithm on $G$ with at most a factor $2C'R+3\Delta'$ increase in round complexity and factor $3\Delta'$ increase in energy complexity.
\end{lemma}

\begin{proof}
    Let $\mathcal{A}$ be a $\CONGEST$ algorithm in $G'$. In a single round of $\mathcal{A}$ a participating vertex $[v]$ with cluster center $v$ may wish to send up to $\Delta'$ messages, one to each of its neighbors, and receive up to $\Delta'$ messages, one from each of its neighbors.  To simulate this in $G$, the messages are computed at the cluster center $v$. The messages to be sent are disseminated within $[v]$ via a $\DOWNCAST$ operation, and sent over to neighboring clusters via an $\INTERCAST$ operation. The messages to be received are sent over from neighboring clusters via the $\INTERCAST$ and then collected at the cluster center $v$ via an $\UPCAST$ operation. 
    If only a single message were to be sent and a single message received, then $\DOWNCAST$ and $\UPCAST$ would simply be flooding and echo in the cluster respectively, implemented with sleeping in order to save energy. And $\INTERCAST$ could be achieved by the vertices on the boundary of the cluster sending the message to all their neighbors. However, since up to $\Delta'$ messages are to be sent and received, we need to describe these operations more carefully.

    Note that constant distortion $C$ at scale $R$ means that all the clusters have diameter at most $(C-1)R$. Consider a time interval $I$ starting at time $t_0$ during which a single round $j$ of the algorithm on $G'$ is to be simulated. If a vertex in $G'$ is meant to sleep in round $j$ then all of the vertices in the corresponding cluster will sleep during interval $I$. Otherwise, the vertices will participate in the $\DOWNCAST$, $\INTERCAST$, and $\UPCAST$. To describe these, we will assume at first that the $\ID$s of the cluster centers of neighboring clusters are known. We will remove this assumption momentarily. 

    $\DOWNCAST$. If $[v]$ is not meant to be sleeping in round $j$, then $v$ prepares the $\Delta'$ messages, tagging each with the $\ID$ of its recipient. It then floods these into its cluster in consecutive rounds $t_0 +1, t_0 +2, \dots, t_0 +\Delta'$. Each vertex $w$ in the cluster knows its depth $\delta_w$ in the cluster, so it wakes up in round $t_0 + \delta_w$ and receives the first message from its parent. In each of the following $\Delta' - 1$ rounds, it receives another message from its parent in the tree, while also forwarding the message from the previous round to its children in the tree. Finally in round   $t_0 + \delta_w +\Delta'$ it forwards the last message to all its children in the tree and goes back to sleep. Since the cluster diameter is at most $CR$, by round  $t_0+CR+\Delta'$ everyone in the cluster has all $\Delta'$ messages and $\DOWNCAST$ is done.
    It takes $CR+\Delta'$ rounds and $\Delta'$ units of energy per vertex. 

    $\INTERCAST$. All the vertices that have at least one neighbor in a different cluster wake up for $\INTERCAST$, which occurs in rounds $t_0 + CR + \Delta' + 1, t_0 + CR + \Delta' + 2, \dots, t_0 + CR+ 2\Delta'$. During these rounds the vertices send the $\Delta'$ messages, one per round,  to all their neighbors, and receive messages from those neighbors. Now, during this time, a vertex $w$ may receive as many as $\deg(w) \Delta'$ different messages. However, at most $\Delta'$ of these are tagged for its own cluster. It stores these for sending during $\UPCAST$ and discards the rest. Note that $w$ might not get messages from all the neighbouring clusters.  $\INTERCAST$ takes $\Delta'$ rounds and energy.

    $\UPCAST$ begins in round $t_0 + CR + 2\Delta'+1$ and continues for $CR+\Delta'$ rounds. Let $T =t_0 + 2CR + 3\Delta'$.
    A vertex $w$ in the cluster at depth $\delta_w$ wakes up for rounds the $\Delta' +1$ rounds numbered. $T-\delta_w - \Delta',  T-\delta_w - \Delta' +1, \dots, T-\delta_w$. In the first $\Delta'$ of these, it receives up to $\Delta'$ messages from its children. On the last $\Delta'$   of them, its sends up to $\Delta'$ messages to its parent. Even though it is receiving up to $\Delta'$ messages from each child, Since there are at most $\Delta'$ distinct messages from neighboring clusters in total,   $\Delta'$ rounds are enough to forward them up.  Each of the messages from a neighboring cluster has been received by some leaf and therefore has some path to the cluster center. Thus by round $T$ the cluster center $v$ has received all the messages from neighboring clusters. 

    Finally we note that although we assumed that the $\ID$s of neighboring clusters are known, this assumption is easy to remove. Indeed, the $\ID$s can be shared using $\DOWNCAST$, $\INTERCAST$ and $\UPCAST$!
    First the cluster centers use $\DOWNCAST$ to tell their cluster their $\ID$s. Then $\INTERCAST$ is used to send the $\ID$ to neighboring clusters. Since the same message is going to all neighbors, it does not need to be tagged with their $\ID$s. Finally, $\UPCAST$ is used to collect the $\ID$s of all the neighbors at the cluster center. This only needs to be done once, and then the $\ID$s are known. 

    Thus we have shown that one round of a $\CONGEST$ algorithm on the cluster graph $G'$ can be simulated in $G$ with a $2CR +3\Delta'$ factor increase in round complexity and a $3\Delta'$ factor increase in energy complexity.
\end{proof}

For the remainder of the section, we will assume that the underlying communication graph $G$ has uniformly bounded independence with parameters $\gamma$ and $\ddim$.

Lemmas~\ref{lem:CONG-sim-shortpath}  and~\ref{lem:CONG-sim} will be useful in building a hierarchical clustering. Recall that when doing this in the $\LOCAL$ model, we used the level $i$ clusters to save energy while selecting a new set of cluster centers from $G$ for the next level. The bandwidth limitations pose challenges that prevent us from directly 
doing the same thing in $\CONGEST$. Instead we will choose the new set of cluster centers to be a maximal subset of the existing cluster centers. It turns out that in this case, the chosen centers do not form a maximal independent set in the power graph. However they are a \emph{ruling set} (defined below) which is good enough for our purpose.

The following notion of \emph{ruling set} dates back at least to~\cite{awerbuch1989network};
see also~\cite{barenboim2016locality}.
\begin{definition}
    Let $G = (V,E)$ be a graph, let $U \subseteq V$, and let $\alpha, \beta \ge 1$.  We say 
    $S \subseteq U$ is an $(\alpha,\beta)$-ruling set for $U$ with respect to $G$, if, for
    every $v \in U$, $d_G(v,S) \le \beta$ and if for every $v \ne w \in S$, $d_G(v,w) \ge \alpha.$
\end{definition}

\begin{lemma} \label{lem:ruling-MIS}
Suppose $S$ is an $(\alpha,\beta)$-ruling set for $V$ with respect to $G$, and let $\gamma \ge \alpha$.
Let $S'$ be a subset of $S$, maximal with respect to being a subset  of $S$ that is also  an independent set in $G^{\le \gamma}$.
Then $S'$ is a $(\gamma+1,\gamma+\beta)$-ruling set for $V$ with respect to $G$.
\end{lemma}

\begin{proof}
This is an immediate consequence of the definitions and the triangle inequality
(We do not need the assumption $\gamma \ge \alpha$, but when $\gamma < \alpha$,
the result is trivial: in this case
we have $S'= S$, which is already known to be a ruling set with the better parameters $(\alpha,\beta).$)
\end{proof}

Our hierarchical clustering algorithm works as follows.   Start with $S_0 = V$.  Recursively, for $i \ge 0$,
let $S_{i+1}$ be a subset of $S_i$,  maximal with respect to being a subset  of $S_i$ that is also  an 
independent set in $G^{\le  2^i}$.  Applying  Lemma~\ref{lem:ruling-MIS} we find by induction that $S_i$ is a
$(2^{i-1}+1,2^{i}-1)$-ruling set for $V$ with respect to $G$. However, we need to show that it is possible to construct these sets $S_i$ in low-energy $\CONGEST$. 

Note that $S_1$ is just any maximal independent set in $G$ and can be found using Luby's algorithm, and the corresponding Voronoi clusters are contained in the ball of radius 1 around the vertices in $S_1$ and have diameter at most 2. These clusters are formed in one step by each vertex in $S_1$ recruiting all its neighbors, and each vertex joining the cluster of its highest $\ID$ neighbor in $S_1$.

We will construct the next levels inductively. Suppose $S_i$ and its Voronoi clusters have already been constructed. Since $S_i$ is a $(2^{i-1}+1,2^{i}-1)$-ruling set for $V$ with respect to $G$, each cluster is contained in a ball of radius $2^{i}-1$ around its center.

Let $G_i$ be the cluster graph relative to Voronoi clusters centered in $S_i$.

\begin{lemma}\label{lem:gi-clust}
    If $G$ is a graph with uniformly bounded independence, then $G_i$ satisfies:
    \begin{itemize}
        \item $G_i$ has constant maximum degree
        \item $G_i$ has constant distance distortion at scale $2^i$.  
    \end{itemize}    
\end{lemma}

\begin{proof}
    The fact that $G_i$ has constant maximum degree follows immediately from $G$ having uniformly bounded independence, since by the triangle inequality, centers of adjacent clusters are within distance $2^i$ of each other, but the ball of radius $2^i$ around a center is its 2-neighborhood in $G^{\le 2^{i-1}}$ and can contain a $G^{\le 2^{i-1}}$-independent set at most $\gamma 2^{\ddim}$.

    To see that $G_i$ has constant distance distortion at scale $2^i$, we note that since $S_i$ is a$(2^{i-1}+1,2^{i}-1)$-ruling set for $V$, we can mimic the proof of Lemma~\ref{lem:distortion-plain} to show that 
    \[
    \left\lceil \frac{d(v,w)+1}{2(2^i -1)+1} \right\rceil \le d'([v],[w])+1 \le 
    \left\lceil \frac{d(v,w)+1}{2(2^{i-1})+1} \right\rceil \gamma 2^k,
    \]
    where $d$ denotes shortest path distance in $G$, and $d'$ denotes shortest path distance in  $G_i$.
    Rearranging terms, we see that $G_i$ has constant distance distortion at scale $2^i$
\end{proof}

Lemma~\ref{lem:gi-clust} implies that $G_i$ satisfies the hypotheses for Lemma~\ref{lem:CONG-sim}, and therefore any $\CONGEST$ algorithm on $G_i$  can be simulated on the underlying graph G for only a constant factor increase in energy complexity, and a $C 2^i$ factor increase in round complexity.

Now let $G'_i$ be the graph with the same vertex set as $G_i$ but with adjacencies inherited from $G^{\le 2^{i}}$. That is for $v, w \in S_i$, we will consider $[v]$ and $[w]$ adjacent in $G'_i$ if $d_G(v, w) \le 2^i$.

\begin{lemma}
    $G'_i$  has constant maximum degree.
\end{lemma}
\begin{proof}
Since $S_i$ is a $(2^{i-1}+1,2^{i}-1)$-ruling set for $V$ with respect to $G$, any two vertices in $S_i$ are distance at least $2^{i-1}+1$ apart, \emph{i.e.,} that $S_i$ is an independent set in $G^{\le 2^{i-1}}$. Since $G$ has uniformly bounded independence, $G^{\le 2^{i-1}}$ has bounded independence with the same parameters $\gamma$ and $\ddim$. Thus for any $v \in S_i$, $\nbd{2^i}{v}$, which is a the 2-neighbourhood of $v$ in $G^{\le 2^{i-1}}$ contains at most $\gamma 2^{\ddim}$ vertices from $S_i$. Thus, the maximum degree of $G'_i$ is 
$\gamma 2^{\ddim} = O(1)$.
\end{proof}

\begin{lemma}
    Suppose the hierarchical multi-scale clustering on $G$ has been built to level $i$ and $S_i$ is the set of cluster centers at level $i$. Then the adjacencies in $G'_i$ can be computed using $O(i)$ energy.
\end{lemma}
\begin{proof}
    This is accomplished using the $\SAF$ algorithm to simulate BFS to depth $2^i$ from each vertex in $S_i$. The simulation uses $\WARN$ to wake vertices up at the appropriate time to participate in the BFS waves. This costs $O(1)$ energy per level, So that the overall energy cost of the algorithm is $O(i)$.
\end{proof}

\begin{lemma}
If $v, w \in S_i$ are such that $[v]$ and $[w]$ are adjacent in $G'_i$ then \[
        d_{G_i}([v], [w])\le  \gamma 2^{\ddim} = O(1)
        \]
\end{lemma}
\begin{proof}
    For $v, w \in S_i$, $[v]$ and $[w]$ being adjacent in $G'_i$ means that $d_G(v, w) \le 2^i$. Let $u$ be the midpoint of the shortest path between $v$ and $w$. Consider any $x$ on the shortest path from $v$ to $w$. Then $d_G(x, u) \le 2^{i-1}$. Also, $x$ is within distance $2^{i-1}$ of one of $v$ and $w$. It follows that $x$ is within distance $2^{i-1}$ of its own cluster center in $S_i$. By the triangle inequality, the cluster center of $x$ is within distance $2^i$ of $u$. But $\nbd{2^i}{u}$ is the 2-neighborhood of $u$ in $G^{\le 2^{i-1}}$ and therefore, by uniformly bounded independence of $G$, contains at most $\gamma 2^{\ddim}$ vertices from $S_i$. Thus, at most  $\gamma 2^{\ddim}$ clusters of the Voronoi clustering centered at $S_i$ can intersect the shortest path between $v$ and $w$. It follows that the distance between $[v]$ and $[w]$ in the cluster graph $G_i$ is at most $\gamma 2^{\ddim} = O(1)$.
\end{proof}

Combining the last three lemmas, we see that $G'_i$ satisfies the conditions of Lemma~\ref{lem:CONG-sim-shortpath}, enabling any $\CONGEST$ algorithm on $G'_i$ to be simulated on $G_i$ for a constant factor increase in round and energy complexities. But we already observed that algorithms on $G_i$  can be simulated on the underlying graph G for only a constant factor increase in energy, and a $C 2^i$ factor increase in the number of rounds. Composing these, we see that algorithms on $G'_i$ can be simulated on the underlying graph $G$ for a constant factor increase in energy and $O(2^i)$ factor increase in the number of rounds.

Now we can simulate Luby's algorithm to construct a maximal independent set in $G'_i$. As we have just seen, 
this can be done in $O(2^i \log n)$ rounds and $O(\log n)$ energy. 
Let $S_{i+1}$ be the cluster centers of the clusters selected in the MIS. Then $S_{i+1}$ is an independent set in $G^{\le 2^i}$ that is maximal in $S_i$, and this is exactly what we were trying to construct.   

Bootstrapping in this fashion, we can construct the sets $S_i$ and corresponding cluster graphs all the way up to $i = \log D$ where $D$ is the diameter of $G$.   
Once this is completed, the precondition mentioned in Theorem~\ref{thm:main-congest} will be met, completing the proof.

\subsection{\texorpdfstring{$\RADIO$}{Radio-CONGEST}}

In the $\RADIO$ setting, the main differences are 
\begin{enumerate}
    \item vertices can send at most one message per round, and
    \item if two or more vertices adjacent to vertex $v$ send in the same round, there is a ``collision''  at $v$, and $v$ receives no message.  
\end{enumerate}
The first difference does not affect the overhead of our low-energy simulation, because the one-message-per-round restriction affects both the
simulated algorithm $\mathcal{A}$ and the simulating algorithm $\mathcal{A}'$ about equally.
To address the second difference, that is, the possibility of collisions, we use the Back-off algorithm to ensure that, with high probability, 
for any round in which at least one neighbor of a listening vertex attempts to send, one of the messages will be successfully received.  Since
this is all that is needed for $\WARN$ to work successfully, the low-energy $\SAF$ simulation algorithm will succeed with high probability.  The 
use of \emph{backoff}---the \emph{decay} algorithm of \cite{bar1992time}---does incur an additional $O(\log(n) \log(\Delta))$ factor overhead to the round and energy complexities.

\section{Approximately Solving Combinatorial Optimization Problems}
 \label{sec:approx}



In this section, we prove \cref{thm:thmApxOpt}.

\thmApxOpt*

\begin{proof}
We first show how to design a randomized algorithm that outputs a feasible solution whose size is at least $(1-\epsilon) \OPT$ \emph{in expectation}, and then we derandomize the algorithm to obtain a desired deterministic algorithm.

\subparagraph{Randomized Algorithm.} 
Let $R = \Theta\left(\frac{1}{\epsilon}\right)$.
We begin by running \cref{alg:voronoi-rand}, which is the low-distortion clustering algorithm of \cref{thm:LOCAL-clustering}. The clustering has the following properties.
\begin{itemize}
    \item The diameter of each cluster is $O(R) = O\left(\frac{1}{\epsilon}\right)$.
    \item Each edge is a crossing edge with probability $O(1/R)$, which we can ensure is at most $\epsilon$ by choosing a sufficiently large constant in $R = \Theta\left(\frac{1}{\epsilon}\right)$.
\end{itemize}
The round complexity and energy complexity of the clustering algorithm are $O\left(\frac{1}{\epsilon} \log^* n\right)$ and $O\left(\log\frac{1}{\epsilon} \log^* n\right)$, respectively.

We ignore all the crossing edges, solve the optimization problem optimally by brute force within each cluster, and take the union of the individual solutions as the final approximate solution. The feasibility of the output is evident. To show that the \emph{expected} size of the solution is at least $(1-\epsilon) \OPT$, in the analysis, we fix $M^\ast \subseteq E$ to be an arbitrary optimal solution. Since each edge is a crossing edge with probability at most $\epsilon$, at least $1-\epsilon$ fraction of $M^\ast$ are not crossing edges. Since our algorithm finds an optimal solution for the graph after removing the crossing edges by brute force, the expected size of the output is at least $(1-\epsilon) |M^\ast| = (1-\epsilon)\OPT$.

To show that our algorithm achieves the desired round and energy complexity, observe that solving the optimization problem optimally by brute force within each cluster costs only $O\left(\frac{1}{\epsilon}\right)$ rounds and $O(1)$ energy by gathering the entire topology of the cluster to the cluster center via the BFS tree, computing an optimal solution locally at the cluster center, and broadcasting the solution via the BFS tree.


\subparagraph{Derandomization.}  While the above guarantee only holds in expectation, as we already saw in the proof of \cref{thm:derand}, \cref{alg:voronoi-rand} can be derandomized to deterministically give a clustering in which at most an $O(R)$ fraction of the edges are crossing edges. Unfortunately, we cannot directly use this derandomized clustering algorithm to get approximate solutions to our optimization problems. This is because the derandomized algorithm of \cref{thm:derand} makes strategic choices of start times to ensure that there are few crossing edges, but the resulting set of edges that become crossing edges may not be the ``right'' ones for the considered optimization problem. Fortunately, instead of derandomizing the clustering algorithm, we can use similar ideas to derandomize the whole approximation algorithm at once, resulting in a \emph{deterministic} algorithm that finds a $(1-\epsilon)$-approximate solution. We describe this in more detail below.
For the sake of simplicity, here we only consider the maximum matching problem. The maximum cut problem can be similarly handled.


    
    We are going to derandomize the randomized algorithm described above by choosing the start times in \cref{alg:voronoi-rand} cleverly instead of randomly via the method of conditional expectations. 
    
    The proof is similar to the proof of \Cref{thm:derand}, so we only highlight the differences. Imagine a cluster center $s$ that is trying to pick its start time for recruiting its cluster.
    We make the following observation: Outside of the ball of radius $2.2 R +1$ around $s$, whether or not an edge eventually becomes an edge in the matching is not affected by $s$'s choice of start time. Therefore, here we still consider the same graph $G^*$ on the cluster centers where there is an edge between two centers in $E^*$ the centers are within distance $2.2R + 1$ of each other.

We plan to set the start time for each cluster center in $V^*$ in \Cref{alg:voronoi-rand} \emph{sequentially}. When we process one cluster center $s \in V^*$, we set its start time deterministically to \emph{maximize} the conditional expected size of the output matching. To do so, $s$ just needs to gather information within a ball of radius $O(R)$, which can be done in $O(R)$ rounds in the $\LOCAL$ model.
Similar to the proof of \Cref{thm:derand}, any independent set $S \subseteq V^\ast$ of $G^\ast$ can set their start times in parallel, allowing us to parallelize the procedure. We use Linial's algorithm~\cite{Linial92} to properly color $G^*$  using $O(1)$ colors in $O(\logstar n)$ rounds in $G^*$, and this can be simulated in the underlying graph $G$ in 
$O(R\logstar n)$ rounds.
Now we go through the color classes sequentially. The centers in the same color class can choose their start times in parallel in $O(R)$ rounds. 
In the end, we obtain a matching whose expected size is \emph{at least} the \emph{expected} size of the matching for the randomized algorithm, which is known to be at least $(1-\epsilon)\OPT$.

The overall round complexity of the deterministic algorithm is $O(R\logstar n) = O\left(\frac{1}{\epsilon} \log^* n\right)$, which is dominated by the SW algorithm in the description of \cref{alg:voronoi-rand} and Linial's coloring algorithm. Using our low-energy simulation of \Cref{thm:localsim}, the energy complexity of the deterministic algorithm can be automatically improved to $O\left(\log\frac{1}{\epsilon} \log^* n\right)$ without affecting the asymptotic round complexity $O\left(\frac{1}{\epsilon} \log^* n\right)$, as required. 
\end{proof}

\section{MPX Algorithm in Bounded Growth Graphs}
\label{sect:mpx}


In this section, we revisit the MPX clustering algorithm~\cite{MPX} and present a variant of it that achieves better guarantees in bounded growth graphs. 
As discussed earlier, using the terminology of \cref{def:AWVC}, for a given scale factor $R$, the MPX clustering is simply an additively weighted Voronoi clustering $\vor(S,W)$ with $S = V$ and $W(v) = \delta_v$, where $\delta_v$ is an exponential random variable with mean $R$. 

To obtain an efficient distributed implementation of the MPX clustering, it is essential to have a \emph{cutoff} of $C R \ln n$ such that any $\delta_v$-value that exceeds $C R \ln n$ will be truncated to $C R \ln n$, where $C > 1$ is some sufficiently large constant~\cite{MPX}. Formally, let $\tilde{\delta}_v = \min\{\delta_v, C R \ln n\}$ denote the result after truncation, and then we use $\{\tilde{\delta}_v\}_{v \in V}$ as the weights for the additively weighted Voronoi clustering: The clustering is computed by letting each vertex $v$ join the cluster $[u]$ of vertex $u$ such that $\dist(u,v) - \tilde{\delta}_v$ is minimized. 
The cutoff ensures that the vertex $u$ that minimizes $\dist(u,v) - \tilde{\delta}_v$ is within the ball of radius $O(R \log n)$ around $v$, so each cluster has diameter $O(R \log n)$. Therefore, the clustering algorithm can be implemented in $O(R \log n)$ rounds in $\CONGEST$.

A key property of the MPX algorithm with a cutoff of $C R \ln n$ is that each edge crosses two distinct clusters with probability $O(1/R) + 1/\poly(n)$. The two terms correspond to the probability of two bad events. Using the memoryless property of exponential distribution, it is possible to show that an edge is a crossing edge with probability at most $1 - e^{-1/R} = O(1/R)$ in the additively weighted Voronoi decomposition with weights \emph{before truncation} $\{\delta_v\}_{v \in V}$. For each vertex $v \in V$, The probability that  $\delta_v \neq \tilde{\delta}_v$ equals $e^{-C \ln n } = n^{-C}$. By a union bound,  with probability at least $1 - n^{1-C}$, the weights $\{\delta_v\}_{v \in V}$ are identical to the weights $\{\tilde{\delta}_v\}_{v \in V}$. The above two bad event probabilities imply the crossing probability $O(1/R) + 1/\poly(n)$. 

The crossing probability upper bound is useful in designing approximation algorithms for various combinatorial optimization problems~\cite{BHKK16,Censor17maxcut,chang2023complexity,FFK21,faour2021approximating} in $\LOCAL$ and $\CONGEST$. For example, if we let each cluster $[u]$ locally compute a maximum matching of the subgraph of $G$ induced by $[u]$. Then the union of all these matchings is a $(1 - O(1/R) - 1/\poly(n))$ approximate maximum matching of $G$ \emph{in expectation}. Therefore, a $(1 - \epsilon)$-approximate maximum matching can be computed in $O(\epsilon^{-1} \log n)$ rounds in the $\LOCAL$ model via the MPX algorithm with $R = \Theta(1/\epsilon)$, with the slight caveat that the approximation ratio only holds in expectation.

\subparagraph{MPX With a Smaller Cutoff.}
The reason to use $\Theta(R \log n)$ as the cutoff is that the $\Theta(\log n)$ factor is needed to do a union bound over all vertices. Intuitively, if we restrict our attention to graphs of bounded growth, then we should be able to use a much smaller cutoff, as for each edge $e$, the number of vertices whose randomness can affect whether $e$ is a crossing edge should be much smaller due to the bounded growth property, so there is no need to do a union bound over all vertices in the graph. In the subsequent discussion, we will make this informal argument precise, and we will show that using $\Theta(R \log R)$ as the cutoff is sufficient to show that an edge crosses two distinct clusters with probability $O(1/R)$ in bounded growth graphs. 

\subparagraph{Analysis.}
In the subsequent discussion, we redefine $\tilde{\delta}_v = \min\{\delta_v, C R \ln R\}$ as the result of truncating $\tilde{\delta}_v$ with the new cutoff $C R \ln R$, where $C$ is some sufficiently large constant to be determined.
Due to the new cutoff $C R \ln R$,  the vertex $u$ that minimizes $\dist(u,v) - \tilde{\delta}_v$ \emph{must be} within the ball of radius $C R \ln R$ around $v$, so each cluster has diameter $O(R \log R)$, and the MPX algorithm can be implemented in $O(R \log R)$ rounds in $\CONGEST$.

From now on, we assume that the input graph $G = (V,E)$ is \emph{bounded growth} in that there are constants $\beta, \ddim >0$ such that for all $v \in V$ and for all $r$, we have $|\nbd{r}{v}| \le \beta r^{\ddim}$. 

\begin{lemma}\label{clm-mpx-1}
For each vertex $u \in V$, with probability $1 - O(1/R)$, that  the weights $\{\delta_v\}_{v \in \nbd{C R \ln R}{u}}$ are identical to the weights $\{\tilde{\delta}_v\}_{v \in \nbd{C R \ln R}{u}}$.
\end{lemma}
\begin{proof}
    For a specific vertex $v \in \nbd{C R \ln R}{u}$, the probability that  $\delta_v \neq \tilde{\delta}_v$ equals $e^{-C \ln R } = R^{-C}$. By a union bound, the probability that $\{\delta_v\}_{v \in \nbd{C R \ln R}{u}}$ are \emph{not} identical to the weights $\{\tilde{\delta}_v\}_{v \in \nbd{C R \ln R}{u}}$ is at most \[|\nbd{C R \ln R}{u}| \cdot R^{-C} \leq \beta {(C R \ln R)}^{\ddim} \cdot R^{-C} = O(1/R),\]
    by selecting $C$ to be a sufficiently large constant. For example, setting $C = \ddim + 1.001$ is enough.
\end{proof}

We rank the vertices in $\nbd{C R \ln R}{u}=\{v_1, v_2, \ldots, v_s\}$ in non-increasing order of $\delta_v - \dist(u,v)$, where $s = |\nbd{C R \ln R}{u}|$. We write $T_i = \dist(u,v_i) - \delta_{v_i}$, so $T_1 \geq T_2 \geq \cdots \geq T_s$.

\begin{lemma}\label{clm-mpx-2}
For each vertex $u \in V$, with probability $1 - O(1/R)$, $T_1 > T_2 + 2$.
\end{lemma}
\begin{proof}
We first reveal the values of $T_2, T_3, \ldots, T_s$ and do the rest of the analysis conditioning on these values.
By the memoryless property of the exponential distribution, $T_1 - T_2 = \max\{0, \dist(u, v_1) - T_2\} + X$, where $X$  is an  exponential random variable with mean $R$. Therefore,
\[\Prob{T_1 > T_2 + 2} \leq \Prob{X > 2} = e^{-2/R} = 1 - O(1/R),\]
as required.
\end{proof}

Combining \cref{clm-mpx-1,clm-mpx-2}, we obtain the desired crossing probability upper bound.

\begin{lemma}\label{clm-mpx-3}
For each edge $e \in E$, with probability $1 - O(1/R)$, both endpoints of $e$ belong to the same cluster.
\end{lemma}
\begin{proof}
Let $e = \{u,v\}$. Applying \cref{clm-mpx-1,clm-mpx-2} to $u$, we infer that with probability $1 - O(1/R)$, the following two good events happen.
\begin{itemize}
\item The weights $\{\delta_v\}_{v \in \nbd{C R \ln R}{u}}$ are identical to the weights $\{\tilde{\delta}_v\}_{v \in \nbd{C R \ln R}{u}}$.
\item $T_1 > T_2 + 2$.
\end{itemize}
These two good events guarantee that $v$ and all its neighbors join the same cluster in the Additively Weighted Voronoi Decomposition with weights $\{\tilde{\delta}_v\}_{v \in V}$. Therefore, both endpoints of $e$ belong to the same cluster with probability $1 - O(1/R)$.
\end{proof}

We summarize the above discussion as a theorem.

\thmMPXboundedgrowth*

Using \cref{thm-mpx-bounded-growth} as a blackbox, we obtain improved distributed approximation algorithms in bounded growth graphs. For example, we obtain the following corollary.

\thmMPXapx*

\begin{proof}
Let $G=(V,E)$ be any bounded growth graph in that there are constants $\beta, \ddim >0$ such that for all $v \in V$ and for  all $r$, we have $|\nbd{r}{v}| \le \beta r^{\ddim}$. In particular, the maximum degree $\Delta$ of $G$ is at most $\beta = O(1)$, so $G$ is a bounded-degree graph.

\subparagraph{Maximum Matching.} Same as the first part of the proof of \cref{thm:thmApxOpt}, our approximate maximum matching works as follows.
We run the algorithm of \cref{thm-mpx-bounded-growth} with $R = \Theta(1/\epsilon)$. After that, we let each cluster locally compute the maximum matching of the subgraph of $G$ induced by the cluster, and then we take the union of all these matchings. It is clear that the algorithm outputs a matching and finishes in $O\left(\frac{1}{\epsilon} \cdot \log \frac{1}{\epsilon} \right)$ rounds in the $\LOCAL$ model.

To analyze the above algorithm, let $M^\ast$ be any maximum matching of $G$. Since $G$ is a bounded-degree graph, we have $|M^\ast| = \Theta(n)$. By \cref{thm-mpx-bounded-growth}, each edge $e \in E$ crosses two distinct clusters with probability $O(1/R)$. We can make this probability at most  $\epsilon/2$ by selecting $R = \Theta(1/\epsilon)$ to be sufficiently large. Therefore, in expectation, at least $(1-\epsilon/2)$ fraction of the edges in $M^\ast$ are intra-cluster edges. This implies that the size of the matching returned by our algorithm is at least $(1-\epsilon/2) \cdot |M^\ast|$ in expectation.

To turn the \emph{in-expectation} approximation guarantee into one that holds \emph{with high probability}, following the proof idea of~\cite{chang2023complexity}, we make use of a Chernoff bound with bounded dependence~\cite{Pem01}. 
Let $X = \sum_{i=1}^t X_i$ be the sum of any $t$ 0-1 random variables $\{X_i\}_{i \in [t]}$ such that each variable $X_i$ is independent of all other variables except for at most  $d$ of them. For any numbers $\mu \geq \Expect{X}$ and $\delta \in (0,1)$, we have
    \[ \Prob{X\geq(1+\delta)\mu} = O(d)\cdot e^{-\Omega(\delta^2\mu / d)}. \]

Let $X = \sum_{e \in M^\ast} X_e$, where $X_e$ is the indicator random variable for the event that $e$ crosses two distinct clusters. Then we have $\Expect{X} \leq (\epsilon/2) \cdot |M^\ast|$, so we may set $\mu = (\epsilon/2) \cdot |M^\ast|$. For each edge $e = \{u,v\}$, variable $X_e$ depends only on the randomness within the vertices in $\nbd{T}{u} \cup \nbd{T}{v}$, where $T = O\left(\frac{1}{\epsilon} \cdot \log \frac{1}{\epsilon} \right)$ is the round complexity of our algorithm. Since $G$ is a bounded growth graph, we may set $d = \epsilon^{-O(1)} = n^{o(1)}$. Therefore, by the above Chernoff bound with bounded dependence with $\delta = 1$, we infer that 
\[\Prob{X\geq \epsilon \cdot |M^\ast|} = O(d) \cdot e^{-\Omega(\mu / d)} = e^{-n^{1 - o(1)}},\] 
as $\mu = \Theta(n)$ because $|M^\ast| = \Theta(n)$. Therefore, our algorithm outputs a $(1-\epsilon)$-approximate maximum matching with probability $1 - e^{-n^{1 - o(1)}}$, which is even better than $1 - 1/\poly(n)$.

\subparagraph{Maximum Independent Set.}  The algorithm for the approximate maximum independent set problem is similar. We run the algorithm of \cref{thm-mpx-bounded-growth} with $R = \Theta(1/\epsilon)$. After that, we let each cluster locally compute the maximum independent set of the subgraph of $G$ induced by the cluster,
and then we take the union of all these independent sets. Finally, we do a post-processing step: For each inter-cluster edge $e$ whose both endpoints are in the set, we remove any one of them.

Let $\alpha(G)$ denote the size of a maximum independent set of $G$. Since $G$ is a bounded-degree graph, we have $\alpha(G) = \Theta(n)$. Here we let $X = \sum_{e \in E} X_e$, where $X_e$ is the indicator random variable for the event that $e$ crosses two distinct clusters, so $X$ equals the number of inter-cluster edges. Clearly, the size of the independent set returned by our algorithm is at least $\alpha(G) - X$, so we just need to show that $X \leq \epsilon \cdot \alpha(G) = \Theta(\epsilon n)$ with probability $1 - 1/\poly(n)$. By selecting $R = \Theta(1/\epsilon)$ to be large enough, we can make sure that $\Expect{X} \leq (\epsilon/2) \cdot \alpha(G)$. Similarly, using the above Chernoff bound with bounded dependence with $\delta = 1$, we infer that $\Prob{X\geq \epsilon \cdot \alpha(G)} = e^{-n^{1 - o(1)}}$, so our algorithm indeed computes a $(1-\epsilon)$-approximate maximum independent set with probability at least $1 - 1/\poly(n)$.
\end{proof}

\subparagraph{Lower Bounds.} To obtain the round complexity $O\left(\frac{1}{\epsilon} \cdot \log \frac{1}{\epsilon} \right)$, it is \emph{necessary} that we restrict our attention to a special graph class, as higher lower bounds were known for general graphs: $(1-\epsilon)$-approximation of maximum independent set requires $\Omega\left(\frac{1}{\epsilon} \cdot \log n \right)$ rounds to compute~\cite{BHKK16,chang2023complexity} and constant-approximation of maximum matching requires $\Omega\left(\min \left\{ \sqrt{\log n / \log \log n}, \log \Delta/ \log \log  \Delta \right\}\right)$ rounds to compute~\cite{coupette2021breezing,KuhnMW16}, where $\Delta$ is the maximum degree of the graph. These lower bounds hold even when the approximation guarantee holds in expectation and apply to the $\LOCAL$ model.

\subparagraph{Remark.}
Although truncating the weights range for the MPX algorithm reduces the maximum diameter from $O(R \log n)$ down to $O(R \log R)$, this does \emph{not} reduce the distortion of the clustering for scale $R$ from $O(\log n)$ to $O(\log R)$.  For example, when $G$ is a cycle, 
a run of $\frac{\log  n}{\log R}$
consecutive single-vertex clusters is not unlikely, which means, for instance, when $R = \Theta(\log n / \log \log n)$, the distortion for scale $R$ is still nearly $\Omega(\log n/\log \log n)$, see \Cref{obs-mpx-2}, whose proof still works even for MPX with a cutoff of $\Theta(R \log R)$.

The reduction in cluster diameter in \cref{thm-mpx-bounded-growth} is still a significant improvement:  MPX clustering with a cutoff of $\Theta(R \log R)$ for scale $X = \Theta\left(R \sqrt{\frac{\log R}{\log n}}\right)$ has distortion $\Theta(\sqrt{\log n \log R})$, due to the rebalancing of underestimated and overestimated distances: $\frac{\Theta(R \log R)}{X} = \frac{X}{\Theta(R / \log n)}$.  This can be used, for example, to remove one factor of $\log n$ from the $\polylog(n)$ energy cost incurred in the constructions of~\cite{dani2022wake} for radio networks.
On the other hand, further improvements to the MPX construction would require at least one additional idea.

\section{Round and Energy Lower Bounds}
\label{sect:lb}

In this section, we prove round and energy lower bounds for the $(1-\epsilon)$-approximate maximum independent set problem in the $\LOCAL$ model.

\thmLBenergy*
\begin{proof}
 Let $n$ be the smallest even number such that $2/n < \epsilon$. Let $G$ be an $n$-vertex cycle. Let $\algo$ be any algorithm that solves the $(1-\epsilon)$-approximate maximum independent set problem in cycles in the $\LOCAL$ model with probability at least $0.99$. Then $\algo$ must compute a maximum independent set of $G$.
 
Let $k$ be an integer to be determined.
We pick two disjoint $(2k+1)$-vertex paths $P_1$ and $P_2$. We write $P_1 = A_1 \circ \{v_1\} \circ B_1$ and  $P_2 = A_2 \circ \{v_2\} \circ B_2$, where $A_1$, $A_2$, $B_1$, and $B_2$ are  $k$-vertex paths. We choose  $P_1$ and $P_2$ in such a way that the distance between $v_1$ and $v_2$ in $G$ is an odd number, so for any maximum independent set of $G$, exactly one of $v_1$ and $v_2$ is in the independent set.
 
 The high-level idea of the proof is to show that if the energy complexity of  $\algo$ is $o\left(\log \frac{1}{\epsilon}\right)$, then with constant probability $v_i$ does not know anything outside of $P_i$, for both $i \in \{1,2\}$. Therefore,  $v_1$ and $v_2$ have to decide independently whether to join the independent set, so $\algo$ fails to output a maximum independent set with constant probability.

 As we are in the $\LOCAL$ model, same as the proof of~\cite[Theorem 1]{chang2018energy}, we may assume, without loss of generality, that the algorithm $\algo$ works as follows.
Every vertex begins in exactly the same state.
Each vertex $v$ locally generates a string $r_v$ of random bits, and afterward,
behaves deterministically.
At any moment in time, each vertex $v$ maintains a connected set of vertices $S$ such that $v$ knows
$r_{u}$ if and only if $u \in S$.  Whenever a vertex $v$ transmits a message, it transmits every useful piece
of information it knows, which is the random string $r_{u}$ for all $u \in S$. 

Following~\cite[Theorem 1]{chang2018energy}, for each subpath $P$ of $G$, let $\mathcal{E}_i[P]$ be the event such that there exists a vertex in $P$ that, after its $i$th wakeup, knows of no information outside $P$, regardless of the choices of random strings outside of $P$. The quantifier over all choices of random strings outside of $P$ implies that $\mathcal{E}_i[P]$ depends only on the randomness inside $P$. It was shown in~\cite[Theorem 1]{chang2018energy} that for each path $P$ of $(13)^i$ vertices, we have
\[\Prob{\mathcal{E}_i[P]} \geq \frac{1}{2}.\]

We select $k$ in such a way that $k = (13)^i$ with $i = o(\log n) = o\left(\log \frac{1}{\epsilon}\right)$ being the energy complexity of $\algo$. For both $j \in \{1,2\}$, let $\mathcal{E}_j^\ast$ be the event where both $\mathcal{E}_i[A_j]$ and $\mathcal{E}_i[B_j]$ occur. Since $\mathcal{E}_i[A_j]$ and $\mathcal{E}_i[B_j]$ are independent events,  for both $j \in \{1,2\}$, we have \[\Prob{\mathcal{E}_j^\ast} =\Prob{\mathcal{E}_i[A_j]} \cdot \Prob{\mathcal{E}_i[B_j]} \geq \frac{1}{4}.\]
Since $\mathcal{E}_1^\ast$ and $\mathcal{E}_2^\ast$ are independent events, we have $\Prob{\mathcal{E}_1^\ast} \cdot \Prob{\mathcal{E}_2^\ast} \geq \frac{1}{4} \cdot \frac{1}{4} = \frac{1}{16}$.

Let $p$ be the probability that $v_1$ joins the independent set conditioning on $\mathcal{E}_1^\ast$. By symmetry, $p$ is also the probability that $v_2$ joins the independent set conditioning on $\mathcal{E}_2^\ast$. Therefore, $\algo$ fails with probability at least
\[\Prob{\mathcal{E}_1^\ast} \cdot \Prob{\mathcal{E}_2^\ast} \cdot (p^2 + (1-p)^2)
\geq \frac{1}{16} \cdot (p^2 + (1-p)^2) =
\frac{1}{16} \cdot \left(2\left(p- \frac{1}{2}\right)^2 + \frac{1}{2}\right) \geq \frac{1}{32} > 0.01,\]
contradicting the assumption that $\algo$ succeeds with probability at least $0.99$.
\end{proof}


\thmLBtime*

The proof of \cref{thm-lb-time-bounded-growth} is obtained by a minor modification to the lower bound proofs in~\cite{BHKK16,chang2023complexity}. In the subsequent discussion, we briefly review their proofs and describe the needed modification. By an indistinguishability argument, it was shown in~\cite{BHKK16} that constant-approximation of maximum independent set needs $\Omega(\log n)$ rounds to solve in the $\LOCAL$ model. Their proof utilizes the Ramanujan graphs constructed in~\cite{lubotzky1988ramanujan}.

\begin{theorem}[\cite{lubotzky1988ramanujan}]\label{thm:ramanujan}
For any two unequal primes $p$ and $q$ congruent to $1 \mod 4$, there exists a $(p+1)$-regular graph $X^{p,q}$ satisfying the following properties.
\begin{description}
\item[Case 1:] $\parens{\dfrac{q}{p}} = -1$.
\begin{itemize}
    \item $X^{p,q}$ is a bipartite graph with $n = q(q^2 - 1)$ vertices.
    \item The girth of $X^{p,q}$ is at least $4 \log_p q - \log_p 4$.
\end{itemize}
\item[Case 2:] $\parens{\dfrac{q}{p}} = 1$.
\begin{itemize}
    \item $X^{p,q}$ is a non-bipartite graph with $n = q(q^2 - 1)/2$ vertices.
    \item The girth of $X^{p,q}$ is at least $2 \log_p q$.
    \item The size of a maximum independent set of $X^{p,q}$ is at most $\frac{2 \sqrt{p}}{p+1} \cdot n$.
\end{itemize}
\end{description}
\end{theorem}

In \cref{thm:ramanujan}, $\parens{\dfrac{q}{p}} = q^{\frac{p-1}{2}} \mod p \in \{-1,0,1\}$ is the Legendre symbol.
For any fixed prime $p$ congruent to $1 \mod 4$, the families of graphs $X^{p,q}$ in the above case 1 and case 2 are infinite. In these graphs with a fixed constant $p$, any $o(\log n)$-round algorithm is not able to distinguish between case 1 and case 2, as the $o(\log n)$-radius neighborhood of each vertex in these graphs are $p$-regular trees. 
Therefore, for any constant $\alpha > 0$, there is a number $p$ such that using an indistinguishability argument for the graphs $X^{p,q}$, it is possible to obtain an $\Omega(\log n)$ lower bound for $\alpha$-approximate maximum independent set~\cite{BHKK16}. 

More specifically, for example, if $p = 17$, then the size of a maximum independent set is $n/2$ in case 1 and is smaller than $0.46 \cdot n$ in case 2. This implies that an $0.92$-approximate maximum independent set requires $\Omega(\log n)$ rounds to compute. Intuitively, the lower bound follows from the fact that such an independent set allows us to distinguish between the two cases for $p = 17$. More formally, by an indistinguishability argument, if an $o(\log n)$-round algorithm computes an independent set of size at least $0.92 \cdot n/2 = 0.46 \cdot n$ in expectation for case 1, then the same algorithm also computes an independent set of size at least $0.46 \cdot n$ in expectation for case 2, which is impossible.

The $\Omega(\log n)$ lower bound can be extended to an $\Omega\left(\frac{1}{\epsilon} \cdot \log n \right)$ lower bound for $(1-\epsilon)$-approximation by replacing each edge in $X^{p,q}$ with a path of odd length $\Theta(1/\epsilon)$~\cite{chang2023complexity,FFK21}. For case 1, the subdivision maintains the bipartiteness property, so the size of a maximum independent set is still $n/2$ after subdivision. For case 2, by selecting a small enough leading constant in $\Theta(1/\epsilon)$, we can make sure that the size of a maximum independent set is smaller than $(1-\epsilon) \cdot n$ in case 2. Hence the desired $\Omega\left(\frac{1}{\epsilon} \cdot \log n \right)$ lower bound for $(1-\epsilon)$-approximation of maximum independent set is obtained~\cite{chang2023complexity}.

Now we prove \cref{thm-lb-time-bounded-growth} by extending this lower bound to bounded growth graphs.

\begin{proof}[Proof of \cref{thm-lb-time-bounded-growth}]
To adapt the lower bound to bounded growth graphs, the only modification needed is to restrict ourselves to the case where $q = \poly(1/\epsilon)$.
This ensures that the graph after the subdivision has bounded growth. As $n = \poly(q) = \poly(1/\epsilon)$, we infer that any algorithm that solves the $(1-\epsilon)$-approximate maximum independent set problem  in expectation requires $\Omega\left(\frac{1}{\epsilon} \cdot \log \frac{1}{\epsilon} \right)$ rounds in bounded growth graphs in the $\LOCAL$ model.   
\end{proof}

\bibliography{refs}

\appendix

\section{Distance Distortion of MPX Clustering}
\label{sec:MPXanalysis}

In this section, we show that the $O(\log n)$ distance distortion bound (\cref{thm:MPX}) for MPX clustering shown in~\cite{Chang20bfs} cannot be improved significantly, even on cycles. Recall that for a given scale factor $R$, the MPX clustering of $G=(V,E)$ is an additively weighted Voronoi clustering $\vor(S,W)$ with $S = V$ and $W(v) = \delta_v$, where $\delta_v$ is an exponential random variable with mean $R$. 

\begin{proposition}\label{obs-mpx-1}
For any $R = O(n^{1-\epsilon})$, where $\epsilon > 0$ is a constant, the MPX clustering on a cycle $G=(V,E)$ with parameter $R$ contains a cluster whose diameter is $\Omega(R \log n)$ with probability $1 - 1/\poly(n)$.
\end{proposition}
\begin{proof}
Let $\ell = \lfloor (\epsilon/2) R \ln n \rfloor$. Given a path $P$ of length exactly $\ell$, let $\mathcal{E}_P$ be the event that the difference between the largest and the second largest values in $\{\delta_v \, : \, v \in P\}$ is at least $\ell$. If $\mathcal{E}_P$ occurs, then there is a cluster whose diameter is at least $(\ell-2)/3 = \Omega(R \log n)$, because the vertices in $P$ can only belong to the following three clusters.
\begin{itemize}
    \item The cluster $[v]$, where $v \in P$  attains the highest value of $\delta_v$ among all vertices in $P$.
    \item The cluster $[u]$, where $u \in V \setminus P$ is a neighboring vertex of the left endpoint of $P$.
    \item The cluster $[w]$, where $w \in V \setminus P$ is a neighboring vertex of the right endpoint of $P$.
\end{itemize}
By the memoryless property of the exponential distribution, the difference between the largest and the second largest values in $\{\delta_v \, : \, v \in P\}$ also follows the exponential distribution with mean $R$, so 
\[\Prob{\mathcal{E}_P} = e^{-\ell/R} \geq n^{-\epsilon/2}.\]
Therefore, with probability at least $n^{-\epsilon/2}$, there is a cluster whose diameter is $\Omega(R \log n)$.
To amplify the probability to $1 - 1/\poly(n)$, we just need to apply a Chernoff bound to any collection of $\Omega\left(n^{\epsilon/2} \log n\right)$ vertex-disjoint paths of length exactly $\ell$, as the events $\mathcal{E}_P$ and $\mathcal{E}_Q$ are independent whenever $P$ and $Q$ are vertex-disjoint.
\end{proof}

\cref{obs-mpx-1} implies the existence of two vertices $v$ and $w$ such that $d(v,w) = \Theta(R \log n)$ and $d'([v],[w]) = 0$, where $d'$ is the distance in the cluster graph. The cluster graph distance \emph{underestimates} the distance between $v$ and $w$ by a distortion factor of $\Theta(\log n)$, as the ``ideal'' distance between $[v]$ and $[w]$ is $\Theta(\log n)$ for scale $R$.

\begin{proposition}\label{obs-mpx-2}
For any $2 \leq R = n^{o(1)}$, the MPX clustering on a cycle $G=(V,E)$ with parameter $R$ contains a sequence of $\Omega(\log n / \log R)$ consecutive single-vertex clusters with probability $1 - 1/\poly(n)$. 
\end{proposition}
\begin{proof}
We fix $\tau = 2R \ln R$. Given a subpath $P$ of the cycle $G$, we consider the following two independent events.
\begin{description}
    \item[Event $\mathcal{E}_P^1$:] $\tau-1 \leq \delta_v < \tau$ for all vertices $v \in P$.
    \item[Event $\mathcal{E}_P^2$:]  $\delta_v < \tau + \dist(v,P) - 1$ for all vertices $v \in V \setminus P$, where $\dist(v,P) = \min_{u \in P} \dist(u,v)$.
\end{description}
If both events occur, then every vertex in $P$ forms a single-vertex cluster. If  $P$ contains at least $\ell = \frac{\log n}{6\log R}$ vertices, then 
\[\Prob{\mathcal{E}_P^1} \geq \left( \Prob{ \Exponential(1/R) \in[\tau-1, \tau)} \right)^\ell
\geq R^{-3 \ell} = \frac{1}{\sqrt{n}}.
\]
Regardless of the length of $P$, we always have
\begin{align*}
 \Prob{\mathcal{E}_P^2} &\geq \left(1 - \sum_{i=1}^\infty \Prob{\Exponential(1/R) > \tau + i-1}\right)^2\\
&= \left(1 - \sum_{i=1}^\infty e^{-\frac{\tau + i-1}{R}}\right)^2\\
&= \left(1 - \frac{1}{R^2} \cdot \sum_{i=1}^\infty e^{-\frac{i-1}{R}}\right)^2\\
&\geq \left(1 - \frac{1}{R(1-1/e)} \right)^2\\
&> 0.01. &(R \geq 2)
\end{align*}
Therefore, by considering any path of at least $\ell$ vertices, we infer that with probability at least $\frac{1}{100 \sqrt{n}}$, there is a sequence of at least $\ell$ single-vertex clusters.

To amplify the probability to $1 - 1/\poly(n)$, we make the following observation: With probability $1 - 1/\poly(n)$, $\delta_v = O(R \log n)$ for all $v \in V$. Therefore, we can afford to consider only the vertices $v \in V \setminus P$ with $\dist(v,P) = O(R \log n)$ in the definition of event $\mathcal{E}_P^2$. With this modification, for any two paths $P$ and $Q$ that are separated by $\Omega(R \log n)$ vertices, the events $\mathcal{E}_P^1$, $\mathcal{E}_P^2$, $\mathcal{E}_Q^1$, and $\mathcal{E}_Q^2$ are independent. Therefore, to amplify the probability to $1 - 1/\poly(n)$, we can apply a Chernoff bound to any collection of $\Omega(\sqrt{n} \log n)$ paths of at least $\ell$ vertices where any two of these paths are separated by $\Omega(R \log n)$ vertices.
\end{proof}

When $R = \Theta(\log n / \log \log n)$, \cref{obs-mpx-2} implies the existence of two vertices $v$ and $w$ such that both $d(v,w)$ and $d'([v],[w])$ are $\Theta(\log n / \log \log n)$, where $d'$ is the distance in the cluster graph. The cluster graph distance \emph{overestimates} the distance between $v$ and $w$ by a distortion factor of $\Theta(\log n / \log \log n)$, as the ``ideal'' distance between $[v]$ and $[w]$ is $\Theta(1)$ for scale $R$.

\section{Graphs With Uniformly Bounded Independence}
\label{appsec:bdi-eg}

In this section, we exhibit some graph classes that have the uniform bounded independence property.

\subsection*{Bounded Growth Graphs}

Recall from \cref{def:bounded_growth} that a graph $G = (V,E)$ is of bounded growth if there are constants $\beta, \ddim >0$ such that for all $v \in V$, for  all $r$,  $|\nbd{r}{v}| \le \beta r^{\ddim}$.
If a graph $G$ has bounded growth, then clearly it also has bounded independence, but we are interested in graphs for which not only $G$ but also all of its power graphs have bounded independence with the same parameters. This does not follow from bounded growth, but does follow from something just a little stronger.

\begin{definition} 
A graph $G = (V,E)$ is of {\bf \emph{strongly} bounded growth} if there are constants $\alpha, \beta, \ddim >0$ such that for all $v \in V$, for all $r>0$,  $\min\{\alpha r^{\ddim}, n\} \le |\nbd{r}{v}| \le \beta r^{\ddim}$.
\end{definition} 
 
\begin{proposition}\label{lem-strongly-bounded-growth}
    A graph $G$ with strongly bounded growth has uniformly bounded independence.
\end{proposition}

\begin{proof}
    Since $G$ has strongly bounded growth, there are constants $\alpha, \beta, \ddim >0$ such that for all $v \in V$, for all $r>0$,
    \[
    \min\{\alpha r^{\ddim}, n\} \le |\nbd{r}{v}| \le \beta r^{\ddim}.
    \]
    Let $\gamma = 3^{\ddim}\beta /\alpha$. We show that $G$ has uniformly bounded independence with parameters $\gamma$ and $\ddim$.
    
    Fix $v \in V$, $R >0$, and $r>0$.  Let $S \subseteq \nbd{Rr}{v}$ be an independent set in $G^{\le R}$.
    Then the balls of radius $R/2$ around vertices in $S$ are pairwise disjoint, \emph{i.e.}, for all $x, y \in S$, $\nbd{R/2}{x} \cap \nbd{R/2}{y} = \emptyset$. By the strong bounded growth property, each such ball has size at least $\min\{\alpha (R/2)^{\ddim}, n\}$, and therefore, 
    \begin{equation}\label{eq:bdgrLB}
    \left| \bigcup_{x \in S} \nbd{R/2}{x} \right| \ge |S| \cdot \min\{\alpha (R/2)^{\ddim}, n\}.
    \end{equation}
    On the other hand, since $S \subseteq \nbd{Rr}{v}$, by the triangle inequality, the set $\bigcup_{x \in S} \nbd{R/2}{x}$ is contained in the ball of radius $R(r+1/2)$ around $v$. Using the upper bound for strong bounded growth and the observation that $r+1/2 \le 3r/2$ as $r\ge 1$ (and also the fact that the total number of vertices is $n$), we see that 
    \begin{equation}\label{eq:bdgrUB}
    \left| \bigcup_{x \in S} \nbd{R/2}{x} \right| \le \min\left\{ \beta R^{\ddim}\left(r+\frac12\right)^{\ddim}, n\right\} \le \min\{ \beta R^{\ddim}(3r/2)^{\ddim}, n\}.
    \end{equation}
    Combining \eqref{eq:bdgrLB} and \eqref{eq:bdgrUB}, we get
    \[
        |S| \, \le \frac{\min\{ \beta R^{\ddim}(3r/2)^{\ddim}, n\}}{\min\{\alpha (R/2)^{\ddim}, n\}}
        \, \le \frac{\beta R^{\ddim}(3r/2)^{\ddim}}{\alpha (R/2)^{\ddim}}
        \, = (3r)^{\ddim} \beta/\alpha 
    \]
    where the second inequality follows from a case analysis of whether $n$ is below, between or above 
    $\alpha (R/2)^{\ddim} < \beta R^{\ddim}(3r/2)^{\ddim}$.
    Setting $\gamma = 3^{\ddim} \beta / \alpha$ completes the proof.
\end{proof}

Thus all graphs with strongly bounded growth have uniformly bounded independence. In particular, this includes paths, cycles, and $k$-dimensional grids and lattices for any constant $k$.

\subsection*{Geometric Graphs} 


Let $G = (V,E)$, where $V$ is a subset of a metric space, and each edge $\{v,w\}$ is present if and only if the (metric) distance between $v$ and $w$ is less than a specified threshold $r$. Then $G$ is called a \emph{geometric graph}. Of particular interest is the case when the metric space in question is $\R^{\ddim}$ with the Euclidean distance. We will call these Euclidean geometric graphs. When $\ddim =2$ and the threshold distance $r = 1$, we get the usual model of \emph{unit disk graphs}.

A popular special case of geometric graphs is when $V$ is chosen somehow randomly, such as by a Poisson point process.  In this case, $G$ is often referred to as a \emph{random geometric graph.}    

\input{comb}

It is easy to see that Euclidean geometric graphs have bounded independence, since the number of Euclidean balls of radius 1/2 needed to cover a Euclidean ball of radius $r$ (centered at a vertex $v$) in $\R^{\ddim}$ is  $\Theta_k(r^k)$,  
and each such ball can contain at most one vertex of an independent set $S$. 
We would like to be able to say that Euclidean geometric graphs also have \emph{uniformly} 
bounded independence, but unfortunately this is not true as evidenced by the \emph{comb graph} (\Cref{fig:comb}). In such a graph $G$, there exists a vertex $v$ whose ball of radius $r$ in $G^{\le R}$ covers $\Theta(rR)$ paths of length $\Theta(rR)$ in $G$, each of which can contribute $\Theta(r)$ vertices to an independent set in $G^{\le R}$, so the ball of radius $r$ around $v$ in $G^{\le R}$ contains an independent set of size $\Theta(Rr^2)$. The dependence on $R$ means that $G$ lacks \emph{uniformly} bounded independence.
The problem arises because there are  vertices that are nearby in Euclidean distance 
but arbitrarily far away in graphical distance in the comb graph.

In contrast, \emph{random} geometric graphs almost surely do not display this behavior, at least when the random process generating the graph has a sufficient rate. 
Specifically, combining \cite[Theorem 4]{DDHM22} for the case of $\ddim = 2$ and the discussion in \cite[Section 5.1]{DDHM22} for higher dimensions, we have the following.

\begin{theorem}[\cite{DDHM22}]
    If $V$ is a uniformly randomly chosen set of $n$ vertices chosen from 
    a $\ddim$-dimensional cube of total volume $n$, and 
    $G$ is the radius-$r$ disk graph, where $r = \omega(\sqrt{\log n})$.  Then, with probability $1 - O(1/n^2)$, Euclidean distances
    equal graphical distance, scaled by a factor $r$, up to rounding and an error factor of $1 + o(1)$.  In particular, for all $v \ne w \in V$
    \[
    d_G(v,w) \le 2 \|v - w\|/r
    \]
\end{theorem}

Note that by appropriate change of scale in the ambient Euclidean space, we may assume that the threshold distance for edges in a geometric graph is $r=1$.

We will show that even when the positions of the vertices of the geometric graph are selected adversarially, as long as there are no large ``holes'' in the ambient space, the graphical distance is bounded by a constant multiple of the Euclidean distance, and this is sufficient to guarantee uniformly bounded independence.

\begin{definition}
Let $G=(V,E)$ be a geometric graph in $\R^{\ddim}$ where for $u,v \in V$, $\{u,v\} \in E$ if $\| u-v\| \le 1$.
We say that $G$ is $\alpha$-dense if the Euclidean balls of radius $1/\alpha$ centered on $V$ cover the convex hull of $V$. Equivalently, every Euclidean ball of radius $1/\alpha$ centered in the convex hull of $V$ contains at least one point of $V$.
\end{definition}

\begin{proposition}
    If a geometric graph $G$ in $\R^{\ddim}$ is $4$-dense,  then for all $u,v \in V$ with $\|u-v\| \, >1$,
    \[
    d_G(u,v) \le 2 \|u-v\|.
    \]
\end{proposition}

\begin{proof}
    Let $u, v \in V$, with $\|u-v\| \, >1$, so that $u$ and $v$ are not adjacent in $G$. Consider a sequence of points $x_0, x_1 \dots x_{\ell}$ on the the line joining $u$ and $v$, such that $x_0$ is at distance 1/4 from $u$ and for $i\ge 1$,  $x_i$ is at distance 1/2 from $x_{i-1}$.
    Here $\ell = \lceil 2\|u-v\| - 1\rceil$ (see \Cref{fig:4dense}).
    All the $x_i$s are in the convex hull of $V$, so by the 4-density of $G$, the Euclidean balls of radius 1/4 around them each contain some point from $V$.  Let $u= w_0, w_1, \dots, w_{\ell-1}, w_{\ell} = v$ be these points. By the triangle inequality, $\|w_i - w_{i-1}\| \le 1$. Thus $(w_0, w_1, \dots, w_{\ell})$ is a path of length $\ell$ from $u$ to $v$ in $G$. It follows that 
    \[
    d_G(u,v) \le \ell \le  \lceil 2\|u-v\| - 1\rceil \le 2\|u-v\|. \qedhere
    \] 
\end{proof}

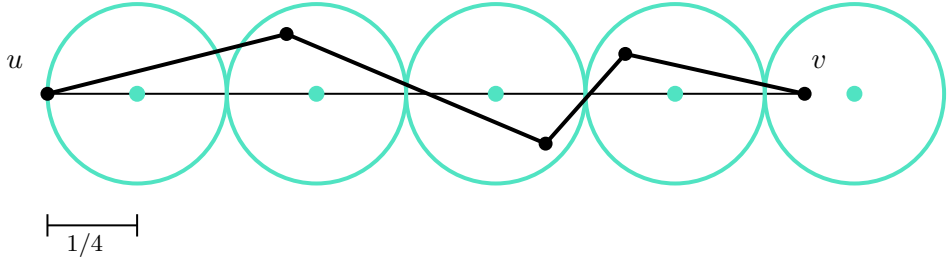
\begin{figure}[ht]
    \centering
\tikzset{every picture/.style={line width=0.75pt}} 

\begin{tikzpicture}[x=0.75pt,y=0.75pt,yscale=-1,xscale=1]

\draw  [color={rgb, 255:red, 0; green, 0; blue, 0 }  ,draw opacity=1 ][fill={rgb, 255:red, 0; green, 0; blue, 0 }  ,fill opacity=1 ] (254,81) .. controls (254,79.34) and (255.34,78) .. (257,78) .. controls (258.66,78) and (260,79.34) .. (260,81) .. controls (260,82.66) and (258.66,84) .. (257,84) .. controls (255.34,84) and (254,82.66) .. (254,81) -- cycle ;
\draw  [color={rgb, 255:red, 0; green, 0; blue, 0 }  ,draw opacity=1 ][fill={rgb, 255:red, 0; green, 0; blue, 0 }  ,fill opacity=1 ] (514,111) .. controls (514,109.34) and (515.34,108) .. (517,108) .. controls (518.66,108) and (520,109.34) .. (520,111) .. controls (520,112.66) and (518.66,114) .. (517,114) .. controls (515.34,114) and (514,112.66) .. (514,111) -- cycle ;
\draw    (137,111) -- (517,111) ;
\draw  [color={rgb, 255:red, 80; green, 227; blue, 194 }  ,draw opacity=1 ][fill={rgb, 255:red, 80; green, 227; blue, 194 }  ,fill opacity=1 ][line width=1.5]  (179,111) .. controls (179,109.34) and (180.34,108) .. (182,108) .. controls (183.66,108) and (185,109.34) .. (185,111) .. controls (185,112.66) and (183.66,114) .. (182,114) .. controls (180.34,114) and (179,112.66) .. (179,111) -- cycle ;
\draw  [color={rgb, 255:red, 80; green, 227; blue, 194 }  ,draw opacity=1 ][line width=1.5]  (137,111) .. controls (137,86.15) and (157.15,66) .. (182,66) .. controls (206.85,66) and (227,86.15) .. (227,111) .. controls (227,135.85) and (206.85,156) .. (182,156) .. controls (157.15,156) and (137,135.85) .. (137,111) -- cycle ;

\draw  [color={rgb, 255:red, 80; green, 227; blue, 194 }  ,draw opacity=1 ][fill={rgb, 255:red, 80; green, 227; blue, 194 }  ,fill opacity=1 ][line width=1.5]  (269,111) .. controls (269,109.34) and (270.34,108) .. (272,108) .. controls (273.66,108) and (275,109.34) .. (275,111) .. controls (275,112.66) and (273.66,114) .. (272,114) .. controls (270.34,114) and (269,112.66) .. (269,111) -- cycle ;
\draw  [color={rgb, 255:red, 80; green, 227; blue, 194 }  ,draw opacity=1 ][line width=1.5]  (227,111) .. controls (227,86.15) and (247.15,66) .. (272,66) .. controls (296.85,66) and (317,86.15) .. (317,111) .. controls (317,135.85) and (296.85,156) .. (272,156) .. controls (247.15,156) and (227,135.85) .. (227,111) -- cycle ;

\draw  [color={rgb, 255:red, 80; green, 227; blue, 194 }  ,draw opacity=1 ][fill={rgb, 255:red, 80; green, 227; blue, 194 }  ,fill opacity=1 ][line width=1.5]  (359,111) .. controls (359,109.34) and (360.34,108) .. (362,108) .. controls (363.66,108) and (365,109.34) .. (365,111) .. controls (365,112.66) and (363.66,114) .. (362,114) .. controls (360.34,114) and (359,112.66) .. (359,111) -- cycle ;
\draw  [color={rgb, 255:red, 80; green, 227; blue, 194 }  ,draw opacity=1 ][line width=1.5]  (317,111) .. controls (317,86.15) and (337.15,66) .. (362,66) .. controls (386.85,66) and (407,86.15) .. (407,111) .. controls (407,135.85) and (386.85,156) .. (362,156) .. controls (337.15,156) and (317,135.85) .. (317,111) -- cycle ;

\draw  [color={rgb, 255:red, 80; green, 227; blue, 194 }  ,draw opacity=1 ][fill={rgb, 255:red, 80; green, 227; blue, 194 }  ,fill opacity=1 ][line width=1.5]  (449,111) .. controls (449,109.34) and (450.34,108) .. (452,108) .. controls (453.66,108) and (455,109.34) .. (455,111) .. controls (455,112.66) and (453.66,114) .. (452,114) .. controls (450.34,114) and (449,112.66) .. (449,111) -- cycle ;
\draw  [color={rgb, 255:red, 80; green, 227; blue, 194 }  ,draw opacity=1 ][line width=1.5]  (407,111) .. controls (407,86.15) and (427.15,66) .. (452,66) .. controls (476.85,66) and (497,86.15) .. (497,111) .. controls (497,135.85) and (476.85,156) .. (452,156) .. controls (427.15,156) and (407,135.85) .. (407,111) -- cycle ;

\draw  [color={rgb, 255:red, 80; green, 227; blue, 194 }  ,draw opacity=1 ][fill={rgb, 255:red, 80; green, 227; blue, 194 }  ,fill opacity=1 ][line width=1.5]  (539,111) .. controls (539,109.34) and (540.34,108) .. (542,108) .. controls (543.66,108) and (545,109.34) .. (545,111) .. controls (545,112.66) and (543.66,114) .. (542,114) .. controls (540.34,114) and (539,112.66) .. (539,111) -- cycle ;
\draw  [color={rgb, 255:red, 80; green, 227; blue, 194 }  ,draw opacity=1 ][line width=1.5]  (497,111) .. controls (497,86.15) and (517.15,66) .. (542,66) .. controls (566.85,66) and (587,86.15) .. (587,111) .. controls (587,135.85) and (566.85,156) .. (542,156) .. controls (517.15,156) and (497,135.85) .. (497,111) -- cycle ;

\draw  [color={rgb, 255:red, 0; green, 0; blue, 0 }  ,draw opacity=1 ][fill={rgb, 255:red, 0; green, 0; blue, 0 }  ,fill opacity=1 ] (134,111) .. controls (134,109.34) and (135.34,108) .. (137,108) .. controls (138.66,108) and (140,109.34) .. (140,111) .. controls (140,112.66) and (138.66,114) .. (137,114) .. controls (135.34,114) and (134,112.66) .. (134,111) -- cycle ;
\draw  [color={rgb, 255:red, 0; green, 0; blue, 0 }  ,draw opacity=1 ][fill={rgb, 255:red, 0; green, 0; blue, 0 }  ,fill opacity=1 ] (384,136) .. controls (384,134.34) and (385.34,133) .. (387,133) .. controls (388.66,133) and (390,134.34) .. (390,136) .. controls (390,137.66) and (388.66,139) .. (387,139) .. controls (385.34,139) and (384,137.66) .. (384,136) -- cycle ;
\draw  [color={rgb, 255:red, 0; green, 0; blue, 0 }  ,draw opacity=1 ][fill={rgb, 255:red, 0; green, 0; blue, 0 }  ,fill opacity=1 ] (424,91) .. controls (424,89.34) and (425.34,88) .. (427,88) .. controls (428.66,88) and (430,89.34) .. (430,91) .. controls (430,92.66) and (428.66,94) .. (427,94) .. controls (425.34,94) and (424,92.66) .. (424,91) -- cycle ;
\draw [line width=1.5]    (137,111) -- (257,81) ;
\draw [line width=1.5]    (257,81) -- (387,136) ;
\draw [line width=1.5]    (387,136) -- (427,91) ;
\draw [line width=1.5]    (427,91) -- (517,111) ;
\draw [color={rgb, 255:red, 0; green, 0; blue, 0 }  ,draw opacity=1 ]   (137,177) -- (182,177) ;
\draw [shift={(182,177)}, rotate = 180] [color={rgb, 255:red, 0; green, 0; blue, 0 }  ,draw opacity=1 ][line width=0.75]    (0,5.59) -- (0,-5.59)   ;
\draw [shift={(137,177)}, rotate = 180] [color={rgb, 255:red, 0; green, 0; blue, 0 }  ,draw opacity=1 ][line width=0.75]    (0,5.59) -- (0,-5.59)   ;

\draw (145,179.4) node [anchor=north west][inner sep=0.75pt]  [font=\footnotesize]  {$1/4$};
\draw (115,90.4) node [anchor=north west][inner sep=0.75pt]    {$u$};
\draw (519,90.4) node [anchor=north west][inner sep=0.75pt]    {$v$};

\end{tikzpicture}
\caption{A sequence of balls of radius $1/4$ covering the line joining $u$ and $v$ in $\R^{\ddim}$. Since each of these must have a vertex in it, the graphical distance is at most twice the Euclidean distance.}
    \label{fig:4dense}
\end{figure}

In two dimensions, we can do a little better, in that we do not need the vertex set to be as dense to control the graphical distances.

\begin{proposition}\label{lem:2d}
    If a unit disk graph $G$ in $\R^2$ is $2\sqrt{2}$-dense, then for all $u,v \in V$ with $\|u-v\| \, >1$,
    \[
    d_G(u,v) \le 2\sqrt{2} \, \|u-v\|.
    \]
\end{proposition}

\begin{proof}
Since $G$ is $2\sqrt{2}$-dense, every ball of radius $\frac{1}{2\sqrt{2}}$ centered in the convex 
hull of $V$ contains a point of $V$. Since such a ball is inscribed in a square of unit diagonal, 
it follows that every square of unit diagonal centered in the convex 
hull of $V$ contains a point of $V$. 
    Let $u, v \in V$, with $\|u-v\| \, >1$, so that $u$ and $v$ are not adjacent in $G$. Consider the line joining $u$ and $v$ (see \Cref{fig:adaptive}). 

In this case, we define a sequence of vertices $u = w_0, w_1, \dots, w_{\ell} = v$ as follows.  Let $w_0 = u$.  Given $w_i$ for $i \ge 0$, define $S_i$ to be the square of unit diagonal, centered on the line from $u$ to $v$, with $w_i$ on its edge perpendicular to $\overline{uv}$.  Let $w_{i+1}$ be a point in $S_i$, whose projection onto the line $\overline{uv}$ is as far towards $v$ as possible.  Then the sequence $w_0, w_1, \dots, w_{\ell}$ is a path in $G$ from $u$ to $v$.  Since, for each $i \ge 0$, $w_{i+2}$ is not in $S_i$, it follows that $w_{i+2}$ is at least $1/\sqrt{2}$ farther along line $\overline{uv}$ than $w_i$, so 
\[\|u-v\| \, \geq \left\lfloor \frac{\ell+1}{2}\right\rfloor \cdot \frac{1}{\sqrt{2}} \geq \frac{\ell}{2\sqrt{2}} \geq \frac{\dist_G(u,v)}{2\sqrt{2}}. \qedhere\]
\end{proof}

 \begin{figure}[ht]
 \centering
 \resizebox{0.9\linewidth}{!}{
 
\tikzset{
pattern size/.store in=\mcSize, 
pattern size = 5pt,
pattern thickness/.store in=\mcThickness, 
pattern thickness = 0.3pt,
pattern radius/.store in=\mcRadius, 
pattern radius = 1pt}
\makeatletter
\pgfutil@ifundefined{pgf@pattern@name@_k398lmzgf}{
\pgfdeclarepatternformonly[\mcThickness,\mcSize]{_k398lmzgf}
{\pgfqpoint{0pt}{-\mcThickness}}
{\pgfpoint{\mcSize}{\mcSize}}
{\pgfpoint{\mcSize}{\mcSize}}
{
\pgfsetcolor{\tikz@pattern@color}
\pgfsetlinewidth{\mcThickness}
\pgfpathmoveto{\pgfqpoint{0pt}{\mcSize}}
\pgfpathlineto{\pgfpoint{\mcSize+\mcThickness}{-\mcThickness}}
\pgfusepath{stroke}
}}
\makeatother
\tikzset{every picture/.style={line width=0.75pt}} 

\begin{tikzpicture}[x=0.75pt,y=0.75pt,yscale=-1,xscale=1]

\draw  [fill={rgb, 255:red, 0; green, 0; blue, 0 }  ,fill opacity=1 ] (94,102) .. controls (94,100.34) and (95.34,99) .. (97,99) .. controls (98.66,99) and (100,100.34) .. (100,102) .. controls (100,103.66) and (98.66,105) .. (97,105) .. controls (95.34,105) and (94,103.66) .. (94,102) -- cycle ;
\draw  [fill={rgb, 255:red, 0; green, 0; blue, 0 }  ,fill opacity=1 ] (594,102) .. controls (594,100.34) and (595.34,99) .. (597,99) .. controls (598.66,99) and (600,100.34) .. (600,102) .. controls (600,103.66) and (598.66,105) .. (597,105) .. controls (595.34,105) and (594,103.66) .. (594,102) -- cycle ;
\draw    (97,102) -- (597,102) ;
\draw  [color={rgb, 255:red, 208; green, 2; blue, 27 }  ,draw opacity=1 ][dash pattern={on 5.63pt off 4.5pt}][line width=1.5]  (187,42) -- (307,42) -- (307,162) -- (187,162) -- cycle ;
\draw  [fill={rgb, 255:red, 0; green, 0; blue, 0 }  ,fill opacity=1 ] (269,140) .. controls (269,138.34) and (270.34,137) .. (272,137) .. controls (273.66,137) and (275,138.34) .. (275,140) .. controls (275,141.66) and (273.66,143) .. (272,143) .. controls (270.34,143) and (269,141.66) .. (269,140) -- cycle ;
\draw  [fill={rgb, 255:red, 0; green, 0; blue, 0 }  ,fill opacity=1 ] (184,60) .. controls (184,58.34) and (185.34,57) .. (187,57) .. controls (188.66,57) and (190,58.34) .. (190,60) .. controls (190,61.66) and (188.66,63) .. (187,63) .. controls (185.34,63) and (184,61.66) .. (184,60) -- cycle ;
\draw  [color={rgb, 255:red, 245; green, 166; blue, 35 }  ,draw opacity=1 ][line width=0.75]  (272,42) -- (392,42) -- (392,162) -- (272,162) -- cycle ;
\draw  [fill={rgb, 255:red, 0; green, 0; blue, 0 }  ,fill opacity=1 ] (334,85) .. controls (334,83.34) and (335.34,82) .. (337,82) .. controls (338.66,82) and (340,83.34) .. (340,85) .. controls (340,86.66) and (338.66,88) .. (337,88) .. controls (335.34,88) and (334,86.66) .. (334,85) -- cycle ;
\draw  [color={rgb, 255:red, 208; green, 2; blue, 27 }  ,draw opacity=1 ][dash pattern={on 5.63pt off 4.5pt}][line width=1.5]  (337,42) -- (457,42) -- (457,162) -- (337,162) -- cycle ;
\draw  [fill={rgb, 255:red, 0; green, 0; blue, 0 }  ,fill opacity=1 ] (399,152) .. controls (399,150.34) and (400.34,149) .. (402,149) .. controls (403.66,149) and (405,150.34) .. (405,152) .. controls (405,153.66) and (403.66,155) .. (402,155) .. controls (400.34,155) and (399,153.66) .. (399,152) -- cycle ;
\draw  [fill={rgb, 255:red, 0; green, 0; blue, 0 }  ,fill opacity=1 ] (497,70) .. controls (497,68.34) and (498.34,67) .. (500,67) .. controls (501.66,67) and (503,68.34) .. (503,70) .. controls (503,71.66) and (501.66,73) .. (500,73) .. controls (498.34,73) and (497,71.66) .. (497,70) -- cycle ;
\draw  [color={rgb, 255:red, 245; green, 166; blue, 35 }  ,draw opacity=1 ] (401,42) -- (521,42) -- (521,162) -- (401,162) -- cycle ;
\draw  [color={rgb, 255:red, 208; green, 2; blue, 27 }  ,draw opacity=1 ][dash pattern={on 5.63pt off 4.5pt}][line width=1.5]  (501,42) -- (621,42) -- (621,162) -- (501,162) -- cycle ;
\draw [color={rgb, 255:red, 74; green, 144; blue, 226 }  ,draw opacity=1 ][line width=1.5]    (97,102) -- (187,60) ;
\draw [color={rgb, 255:red, 74; green, 144; blue, 226 }  ,draw opacity=1 ][line width=1.5]    (187,60) -- (272,140) ;
\draw [color={rgb, 255:red, 74; green, 144; blue, 226 }  ,draw opacity=1 ][line width=1.5]    (337,85) -- (272,140) ;
\draw [color={rgb, 255:red, 74; green, 144; blue, 226 }  ,draw opacity=1 ][line width=1.5]    (337,85) -- (402,152) ;
\draw [color={rgb, 255:red, 74; green, 144; blue, 226 }  ,draw opacity=1 ][line width=1.5]    (500,70) -- (402,152) ;
\draw [color={rgb, 255:red, 74; green, 144; blue, 226 }  ,draw opacity=1 ][line width=1.5]    (500,70) -- (597,102) ;
\draw  [color={rgb, 255:red, 245; green, 166; blue, 35 }  ,draw opacity=1 ][line width=0.75]  (97,42) -- (217,42) -- (217,162) -- (97,162) -- cycle ;

\draw  [fill={rgb, 255:red, 0; green, 0; blue, 0 }  ,fill opacity=1 ] (233,310) .. controls (233,308.34) and (234.34,307) .. (236,307) .. controls (237.66,307) and (239,308.34) .. (239,310) .. controls (239,311.66) and (237.66,313) .. (236,313) .. controls (234.34,313) and (233,311.66) .. (233,310) -- cycle ;
\draw    (236,310) -- (496,310) ;
\draw  [fill={rgb, 255:red, 19; green, 120; blue, 235 }  ,fill opacity=1 ] (303,250) .. controls (303,248.34) and (304.34,247) .. (306,247) .. controls (307.66,247) and (309,248.34) .. (309,250) .. controls (309,251.66) and (307.66,253) .. (306,253) .. controls (304.34,253) and (303,251.66) .. (303,250) -- cycle ;
\draw  [color={rgb, 255:red, 128; green, 128; blue, 128 }  ,draw opacity=1 ][fill={rgb, 255:red, 128; green, 128; blue, 128 }  ,fill opacity=1 ] (259,332) .. controls (259,330.34) and (260.34,329) .. (262,329) .. controls (263.66,329) and (265,330.34) .. (265,332) .. controls (265,333.66) and (263.66,335) .. (262,335) .. controls (260.34,335) and (259,333.66) .. (259,332) -- cycle ;
\draw  [color={rgb, 255:red, 208; green, 2; blue, 27 }  ,draw opacity=1 ][dash pattern={on 5.63pt off 4.5pt}][line width=1.5]  (306,235) -- (456,235) -- (456,385) -- (306,385) -- cycle ;
\draw  [fill={rgb, 255:red, 0; green, 0; blue, 0 }  ,fill opacity=1 ] (396,365) .. controls (396,363.34) and (397.34,362) .. (399,362) .. controls (400.66,362) and (402,363.34) .. (402,365) .. controls (402,366.66) and (400.66,368) .. (399,368) .. controls (397.34,368) and (396,366.66) .. (396,365) -- cycle ;
\draw  [color={rgb, 255:red, 245; green, 166; blue, 35 }  ,draw opacity=1 ] (236,235) -- (386,235) -- (386,385) -- (236,385) -- cycle ;
\draw  [color={rgb, 255:red, 128; green, 128; blue, 128 }  ,draw opacity=1 ][fill={rgb, 255:red, 128; green, 128; blue, 128 }  ,fill opacity=1 ] (279,293) .. controls (279,291.34) and (280.34,290) .. (282,290) .. controls (283.66,290) and (285,291.34) .. (285,293) .. controls (285,294.66) and (283.66,296) .. (282,296) .. controls (280.34,296) and (279,294.66) .. (279,293) -- cycle ;
\draw  [draw opacity=0][pattern=_k398lmzgf,pattern size=6pt,pattern thickness=0.75pt,pattern radius=0pt, pattern color={rgb, 255:red, 155; green, 155; blue, 155}] (306,234.5) -- (386,234.5) -- (386,385) -- (306,385) -- cycle ;
\draw [color={rgb, 255:red, 4; green, 114; blue, 250 }  ,draw opacity=1 ][line width=1.5]    (306,250) -- (236,310) ;
\draw [color={rgb, 255:red, 4; green, 118; blue, 248 }  ,draw opacity=1 ][line width=1.5]    (306,250) -- (399,365) ;

\draw (206,293.4) node [anchor=north west][inner sep=0.75pt]  [font=\large]  {$u$};
\draw (76,102.4) node [anchor=north west][inner sep=0.75pt]  [font=\large]  {$u$};
\draw (599,105.4) node [anchor=north west][inner sep=0.75pt]  [font=\large]  {$v$};

\end{tikzpicture}}
\caption{Illustration of the proof of \Cref{lem:2d}. In two dimensions we can construct the path adaptively, using squares with unit diagonal, each of which starts just past the previous point on the path. Since the next point cannot lie in the shaded region, any two consecutive steps are guaranteed to make at least one square's width of progress in the horizontal direction.}
    \label{fig:adaptive}
\end{figure}
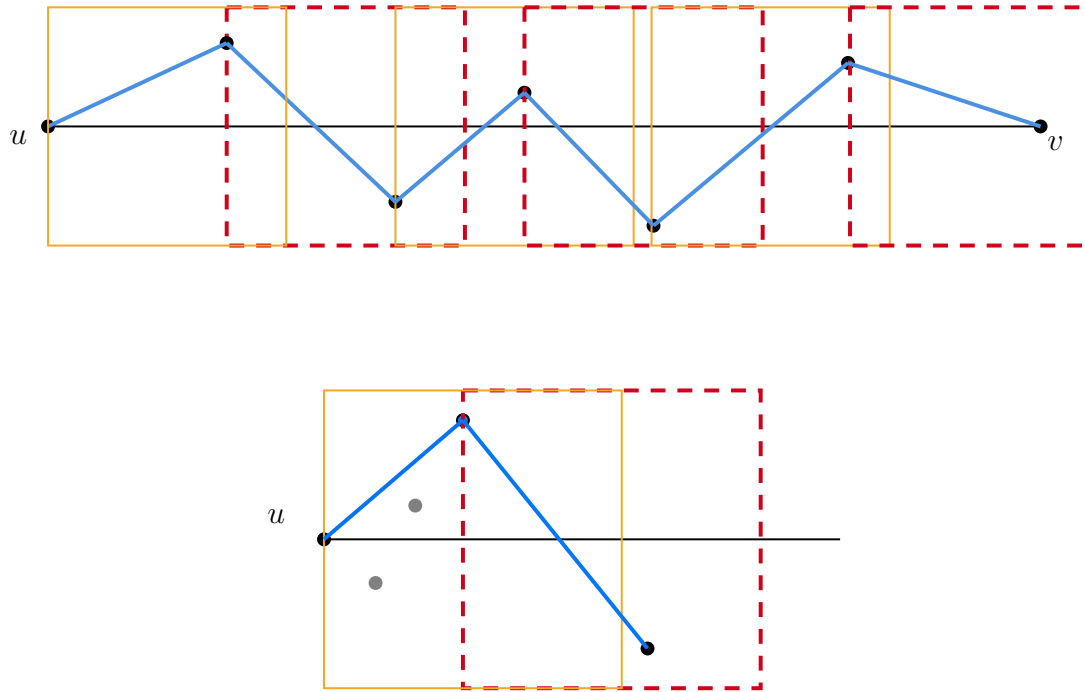


\begin{proposition}
    Let $G$ be a geometric graph in $\R^{\ddim}$, such that for all non-adjacent $u,v \in V$, 
    \[d_G(u,v) \le C \|u-v\| \]
    for some constant $C > 0$. Then $G$ has uniformly bounded independence. 
\end{proposition}

\begin{proof}
Fix $R$ and $u \in V$. 
    suppose $S \subseteq B_{rR}(u)$ is an independent set of $G^{\le R}$.  
    Then for all $v, w \in S$, $d_G(v,w) > R$. By hypothesis, it follows that for all $v, w \in S$, $\|v-w\| \, > R/C$. Thus the Euclidean balls of radius $R/(2C)$
    centered at vertices of $S$ are pairwise disjoint.  Moreover, these are all contained within the Euclidean ball of radius  $\left(r+\frac{1}{2C}\right)R$ around $u$ (by the triangle inequality, and since for any geometric graph, the Euclidean distance is at most the graphical distance.)
    It follows that
    \[
    |S| \mathcal{V}_k \left(\frac{R}{2C}\right) ^{\ddim} \le \mathcal{V}_k \left(r+\frac{1}{2C}\right)^{\ddim} R^{\ddim},
    \]
where $\mathcal{V}_k$ is the volume of the unit (Euclidean) ball in $\R^{\ddim}$.
    Thus $|S| \, \le (2Cr+1)^{\ddim}$. Since the size of $S$ grows polynomially in $r$, but does not depend on $R$, we have established that $G$ has uniformly bounded independence.
\end{proof}

\end{document}